\documentclass[aoas]{imsart}

\RequirePackage[OT1]{fontenc}
\RequirePackage{amsthm,amsmath}
\RequirePackage{natbib}
\RequirePackage[colorlinks,citecolor=blue,urlcolor=blue]{hyperref}

\usepackage{amssymb,color}
\usepackage{graphicx}
\PassOptionsToPackage{hyphens}{url}\usepackage{hyperref}
\usepackage[linesnumbered, ruled]{algorithm2e}
\usepackage{bm}
\usepackage[font=small]{caption}
\usepackage{listings}

\arxiv{arXiv:1706.09375}

\startlocaldefs

\newcommand{\br}{{\mathbb R}}

\newcommand\independent{\protect\mathpalette{\protect\independenT}{\perp}}
\def\independenT#1#2{\mathrel{\rlap{$#1#2$}\mkern2mu{#1#2}}}

\newtheorem{proposition}{Proposition}[section]

\newtheorem{lemma}{Lemma}

\newtheorem{theorem}{Theorem}
\newtheorem{remark}[proposition]{Remark} 

\newcommand{\eps}{{\epsilon}}

\newcommand{\norm}[1]{\left \lVert #1 \right \rVert}

\endlocaldefs

\begin{document}

\begin{frontmatter}
	\title{{Multilayer Knockoff Filter: \\ 
		%Multilayer FDR Control for Variable Selection}
		Controlled variable selection \\at multiple resolutions}}
	\runtitle{Multilayer knockoff filter}
	
	\begin{aug}
		\author{\fnms{Eugene} \snm{Katsevich}\thanksref{t1,m1}\ead[label=e1]{ekatsevi@stanford.edu}} \and
		\author{\fnms{Chiara} \snm{Sabatti}\thanksref{t2, m1}\ead[label=e2]{sabatti@stanford.edu}}
		
		\thankstext{t1}{E.K. was supported by the Fannie and John Hertz Foundation and the National Defense Science and Engineering Graduate Fellowship.}
		\thankstext{t2}{C.S. was supported by NIH grants HG006695 and HL113315 and NSF grant DMS 1712800.}
		\runauthor{E. Katsevich and C. Sabatti}
		
		\affiliation{Stanford University\thanksmark{m1}}
		
		\address{Department of Statistics\\
			Stanford University\\
			\printead{e1}\\
			\phantom{E-mail:\ }}
		
		\address{Departments of Statistics\\
			Stanford University\\
			\printead{e2}\\
			\phantom{E-mail:\ }}

	\end{aug}
	
	\begin{abstract}
		We tackle the problem of selecting from among a large number of variables those that are ``important" for an outcome. We consider situations where groups of variables are also of interest. For example, each variable might be a genetic polymorphism and we might want to study how a trait depends on variability in genes, segments of DNA that typically contain multiple such polymorphisms. % Or, variables might quantify various aspects of the functioning of individual internet servers owned by a company, and we might be interested in assessing the importance of each server as a whole on the average download speed for the company's customers. 
		In this context, to discover that a variable is relevant for the outcome  implies discovering that the larger entity it represents is also important. To guarantee meaningful  results with high chance of replicability, we suggest controlling the rate of false discoveries for findings at the level of individual variables and at the level of groups. Building on the knockoff construction of Barber and Cand\`es (2015) and the multilayer testing framework of Barber and Ramdas (2016), we introduce the multilayer knockoff filter (MKF). We prove that MKF simultaneously controls the FDR at each resolution and use simulations to show that it incurs little power loss compared to methods that provide guarantees only for the discoveries of individual variables. We apply MKF to analyze a genetic dataset and find that it successfully reduces the number of false  gene discoveries without a significant reduction in power.
	\end{abstract}
		
\end{frontmatter}

% FOR RESUBMISSION

\section{Introduction}

\subsection{A motivating example}
During the last twenty years, the biotechnology that allows us to identify the locations where the genome of an individual is different from a reference sequence has experienced a dramatic increase in speed and decrease in costs. Scientists have used the resulting wealth of information to investigate empirically how variations in our DNA translate into different measurable phenotypes. While we still know little about the causal mechanisms behind many traits, geneticists agree on the usefulness of a multivariate (generalized) linear model to capture at least as a first approximation the nature of the relation between genetic variation and complex  phenotypes. If $\bm{y} \in \br^{n \times 1}$ is the vector collecting the values of a quantitative trait in $n$ subjects, and $\bm{X} \in \br^{n \times N}$ the matrix storing, column by column, their genotypes at $N$ polymorphic sites in the genome, a starting model for their relation is 
\begin{equation*}
\bm{y} = \bm{X}\bm{\beta} + \bm{\eps},
\end{equation*}
where the coefficients $\bm{\beta} \in \br^{N \times 1}$ represent the contributions of measured genetic variations to the trait of interest. We remark on a few characteristics of this  %set-up its 
motivating genetic application.
(1) The adjective ``complex" referring to a trait is to be interpreted as non-Mendelian, that is influenced by  many different genetic variants: we expect {\em several of the elements of $\bm \beta$ to be nonzero} and we can exploit this fact using multiple regression models.
(2) The main goal of these studies is the identification of which $\beta_j \neq 0$. In other words, the focus is not on developing a predictive model for $\bm{y}$, but on {\em selecting important variables} that represent the biological mechanism behind the trait and whose relevance can be observed across multiple datasets. (3) Recognizing that $\beta_j \neq 0$ corresponds to {\em scientific discoveries at multiple levels}: each column of $\bm{X}\!$ represents a single genetic variant, but these are organized spatially in meaningful ways and their coefficients also give us information about coarser units of variation. For example, a number of adjacent polymorphisms might all map to the same gene, the portion of DNA coding for a protein. If the coefficient for any of these polymorphisms is different from zero, then we can conclude that the {\em gene is important } for the trait under study. This type of discovery is relevant to advancing our understanding of the biology behind a trait. At the same time, knowing which {\em specific variant} influences the phenotype is also relevant: this is the type of information we need for  precise clinical testing and genetic counseling. 

In summary, an ideal solution would {\em identify important genetic variants accounting for  their interdependence, and provide error control guarantees  for the discovery of both variants and  %of 
genes.} The work in this paper  attempts to achieve this goal. We emphasize that similar problems occur in contexts other than genetics. Modern methods of data acquisition often provide information on an exhaustive collection of possible explanatory variables, even if we know a priori that a large proportion of these are not relevant for our outcome of interest. In such cases, we rely on statistical analysis to identify the important variables in a manner that facilitates replicability of results   while achieving appreciable power.  Replication of findings in separate independent studies is the cornerstone of science and cannot be substituted by a type of statistical analysis. Furthermore, the extent to which results replicate depends not only on how the conclusions were drawn from the original data, but on  the characteristics of the follow-up study: does it have enough power? does it target exactly the same ``population'' of the original investigation? etc... Yet, controlling Type-I error is a necessary step towards replicability, important in order to avoid wasting time and money on confirmatory follow-up studies of spurious findings.  %. While replication of the results requires Note that controlling Type-I error is not sufficient for replicability, since successfully reproducing results also requires high enough power in the follow-up study. However,  
%controlling Type-I error at the discovery stage is . 
Finally, it is often the case that we  measure variables at a very fine resolution, and need to aggregate these measurements for meaningful interpretation. We consider three examples in addition to our primary motivating application.
\paragraph{fMRI studies}
Consider, for example, studies that investigate the role of different brain structures. With functional magnetic resonance imaging (fMRI) we measure on the order of a million voxels at regular time intervals during an experiment. These measurements might then be used in a model of neurocognitive ability. Usually, measurements from nearby voxels are aggregated and mapped to recognizable larger-scale brain structures called regions of interest \citep{P07}. With a time dimension also involved, we can group (voxel, time) pairs spatially or temporally. It then becomes important to make sure that the statistical methods we adopt guarantee reproducibility with respect to each kind of scientifically interpretable finding. 

\paragraph{Multifactor analysis-of-variance problems with survey data} In social science, multiple choice surveys are frequently employed in order to gather information about some characteristics of subjects. The results of these surveys can be viewed as predictor variables for certain outcomes, such as income. Since the different answer choices for a given question are coded as separate dummy variables, it makes sense to study the importance of an entire question by considering all the variables corresponding to the same question. However, if a particular question is discovered to be significantly associated with an outcome, then it might also be of interest to know which answer choices are significant. Hence, an analysis at the level of questions (groups of variables) and answer choices (individual variables) is appropriate \citep{YL06}.

\paragraph{Microbiome analysis} The microbiome (the community of bacteria that populate  specific areas of the human body, such as the gut or mouth) has gained attention recently as an important contributing factor for a variety of human health outcomes. By sequencing the bacterial 16S rRNA gene from a specimen collected from a human habitat, it is possible to quantify the abundances of hundreds of bacterial species. Bacteria, like other living organisms, are organized into hierarchical taxonomies with multiple layers including phylum, class, family, and so on. It is of interest to find associations between health outcomes  and the abundances of different types of bacteria, as described with each layer of the taxonomic hierarchy \citep{SH14}.

\subsection{Statistical challenges}
Having motivated our problem with several applications, we give an overview of the statistical challenges involved and of the tools we will leverage.

\paragraph{Controlled variable selection in high dimensional regression}

In a typical genome wide association study (GWAS), the number of subjects $n$ is on the order of tens of thousands  and the number  $N$ of genetic variants (in this case single nucleotide polymorphisms, or SNPs) is on the order of a million. % In order to obtain 
To provide finite sample guarantees of global error,  geneticists typically analyze the relation between $\bm{X}\!$ and $\bm{y}$ using a series of univariate regressions of $\bm{y}$ on each of the columns of $\bm{X}\!$, obtain the p-values for the corresponding t-tests, and threshold them to achieve family wise error rate (FWER) control. This analysis is at odds with the polygenic nature of the traits and the choice of FWER as a measure of global error makes it difficult to recover a substantial portion of the genetic contribution to the %trait  
phenotype \citep{MetV09}.
Using a multiple regression model for analysis and targeting false discovery rate (FDR) \citep{BH95} are promising alternatives.  

Unfortunately, in a context where $N > n$ these are difficult to implement.
Regularized regression, including the lasso \citep{T96} and various generalizations, e.g. \cite{YL06, FetT10}, have proven to be very versatile tools with nice prediction properties,  but they do not come with model selection guarantees in finite samples (for examples of asymptotic properties see \cite{KF00, NetR09}). Recent years have seen progress on this front. While in general it is difficult to obtain p-values for high-dimensional regression, \cite{JM14} propose a construction that is valid under certain sparsity assumptions. Alternatively, conditional inference after  selection \citep{FST14,TT15} can also be used in this context: the idea is to first reduce dimensionality by a screening method, and then apply the Benjamini Hochberg (BH) procedure to %selection-adjusted 
p-values  that have been adjusted for selection \citep{MetT17}. Other approaches have been proposed that bypass the construction of p-values entirely. SLOPE is  a modification of the lasso procedure which provably controls the FDR under orthogonal design matrices \citep{BetC15} and has been applied to GWAS, allowing for a larger set of discoveries which, at least in the analyzed examples, have shown good replicability properties \citep{BetS17a}. The knockoff filter \citep{BC15, CetL16}---which is based on the construction of artificial variables to act as controls---guarantees FDR control for variable selection in a wide range of settings. We will review the properties of knockoffs in Section \ref{sec:background}, as we will leverage them in our construction.
%A particularly promising approach  is the knockoff filter \cite{BC15, CetL16}, which is based on the construction of artificial ``knockoff"  variables. This method guarantees FDR control for variable selection in a wide range of settings; %, and is applicable to both the low and high dimensional case. 
%reviewed briefly in Section \ref{sec:background}, it will be one ingredient of our proposal.

\paragraph{Controlling the false discovery rate at multiple resolutions}
While the standard GWAS analysis results in the identification of SNPs associated with a trait, geneticists also routinely rely on {\em gene} level tests, based on the signal coming from multiple variants associated to the same coding region (\cite{SH16} is a recent review), as well as other forms of aggregate tests, based on 
pathways (see \cite{WLB07} for example). Unfortunately, each of these approaches represents a distinct analysis of the data and the results they provide are not necessarily consistent with each other: we might find association with a SNP, but not with the gene to which it belongs, or with a gene but with none of the variants typed in it. Moreover, multiple layers of analysis increase the burden of multiple testing, often without this being properly accounted for. Yet, geneticists are very interested in controlling type I errors, as follow-up studies are very time consuming and diagnosis and counseling often need to rely on  association results  before a thorough experimental validation is possible.

In this context, we want to  investigate all the interesting levels of resolution simultaneously (e.g. \cite{ZetLa10}) and in a consistent fashion, providing meaningful error control guarantees for all the findings. The fact that we have chosen FDR as a measure of error rate makes this a non-trivial endeavor:
unlike FWER, FDR is a relative measure of global error and its control depends crucially on how one defines discoveries. A procedure that guarantees FDR control for the discovery of single variants does not guarantee FDR control for the discoveries of genes, as we discuss in Section \ref{sec:problem_setup}. This has been noted before in contexts where there is a well defined spatial relationship between the hypotheses and the discoveries of scientific interest are at coarser resolution than the hypotheses tested (e.g. MRI studies \citep{BH07,P07}, genome scan statistics \citep{SetY11}, eQTL mapping \citep{GTEX}). Proposed solutions to these difficulties vary depending on the scientific context, the definition of relevant discoveries, and the way the space of hypotheses is explored; \cite{Y08,BB14, HetS16, BetS17b} explore hypotheses hierarchically, while  \cite{BR15} and \cite{RetJ17} consider all levels simultaneously.

This last viewpoint---implemented in a multiple testing procedure called the p-filter and reviewed briefly in Section \ref{sec:background}---appears to be particularly well-suited to our context, where we want to rely on multiple regression to simultaneously identify important SNPs and genes.
% Our proposal builds upon the construction in  \cite{BR15}, obtaining a formulation that does not rely on p-values and provides a great deal of flexibility in the choice of the analysis strategy at each resolution of interest.

\subsection{Our contribution}

To tackle problems like the one described in the motivating example, we develop  the \textit{multilayer knockoff filter} (MKF), a first link between multi-resolution testing and model selection approaches. We bring together the innovative ideas of knockoffs \citep{BC15, CetL16} and p-filter \citep{BR15, RetJ17} in a new construction that allows us to select important variables and important groups of variables with FDR control. Our methodology---which requires a novel proof technique---does not rely on p-values and provides a great deal of flexibility in the choice of the analysis strategy at each resolution of interest, leading to promising results in genetic applications. 
  
Section \ref{sec:problem_setup_background} precisely describes the multilayer variable selection problem and reviews the knockoff filter and the p-filter. Section \ref{sec:multilayer_knockoff_filter} introduces the multilayer knockoff filter, formulates its FDR control guarantees, and uses the new framework to provide a new multiple testing strategy. % that is more flexible than p-filter. 
Section \ref{sec:numerical} reports the results of simulations illustrating the FDR control and power of MKF. Section \ref{sec:real_data} summarizes a case study on the genetic bases of HDL cholesterol, where MKF appears to successfully reduce false positives with no substantial power loss. The supplementary material \citep{KS18_supp} contains the proofs of our results and details on our genetic findings.

\section{Problem setup and background}  \label{sec:problem_setup_background}

%We begin by formally setting up the problem in Section \ref{sec:problem_setup}, and then present some relevant background in Section \ref{sec:background}.

\subsection{Controlled multilayer variable selection} \label{sec:problem_setup}

%Suppose we have observations on an outcome $y$ and a large collection of covariates
%$X_1, \ldots, X_n$, and that $y$ actually depends upon a small subset of these variables. To be more precise, suppose 
We need to establish a formal definition for our goal of identifying,  among a large collection of covariates $X_1, \ldots, X_N$, the single variables and the group of variables  that are important for an outcome $Y$ of interest. 
Following \cite{CetL16}, we let $(X_1, \dots, X_N, Y)$ be random variables with joint distribution $F$ and assume that the data $(\bm X, \bm y)$ are $n$ i.i.d. samples from this joint distribution\footnote{Our convention in this paper will be to write vectors and matrices in boldface. The only exception to this rule is that the vector of random variables $X = (X_1, \dots, X_N)$ will be denoted in regular font to distinguish it from the design matrix $\bm X$.}. We then define the \textit{base-level hypotheses} of conditional independence as follows:
\begin{equation}
H_j: Y \independent X_j \ |\ X_{-j}.
\label{hypothesis_no_association}
\end{equation}
Here, $X_{-j} = \{X_1, \dots, X_{j-1}, X_{j+1}, \dots, X_N\}$. This formulation  allows us to define null variables in a very general way.
Specifically, we  consider two sets of assumptions on $(\bm X, \bm y)$:
\begin{itemize}
	\item \textit{Fixed design low-dimensional linear model}. Suppose we are willing to assume the linear model
	\begin{equation*}
	\bm y = \bm X \bm \beta + \bm \eps, \quad \bm \eps \sim N(0, \sigma^2 \bm I_n), 
	\end{equation*}
	where $n \geq N$. In this familiar setting, the null hypothesis $H_j$ of conditional independence is equivalent to $H_j: \beta_j = 0$ unless a degeneracy occurs. The standard inferential framework assumes that $\bm X$ is fixed: this is a special case of our general formulation, conditioning on $\bm X$ and, indeed,   the method we will describe 
 can provide FDR guarantees conditional on $\bm X$. This setting was considered by \cite{BC15}.
	\item \textit{Random design with known distribution}. Alternatively, we might not be willing to assume a parametric form for $\bm y| \bm X$, but instead we might  have access to (or can accurately estimate) the joint distribution of each row of $\bm X$. This corresponds to a shift of the burden of knowledge from the the relation between $\bm y$ and $\bm X$ to the distribution of  $\bm X$.
\end{itemize}

The methodology we develop in this paper works equally well with either assumption on the data. In particular, note that the second setting can accommodate categorical response variables as well as continuous ones. Note that the random design assumption is particularly useful for genetic association studies, where $n < N$ and we can use knowledge of linkage disequilibrium (correlation patterns between variants) to estimate the distribution of $\bm X$ (see, e.g. \cite{SetC17}).

Our goal is to select $\mathcal S \subset [N] = \{1, \dots, N\}$ to approximate the set of non null variables. %Following \cite{BR15}, w
We consider situations when $\mathcal S$ is interpreted at $M$ different ``resolutions" or \textbf{layers} of inference. The $m$th layer corresponds to a partition of $[N]$, denoted by $\{\mathcal A_g^m\}_{g \in [G_m]}$, which divides the set of variables into groups representing units of interest. Our motivating example corresponds to $M = 2$: in the first partition, each SNP is a group of its own and in the second, SNPs are grouped by genes. Other meaningful ways to group SNPs in this context could be by functional units or by chromosomes, resulting in more than two layers. 
We note that the groups we consider in each layer are to be specified in advance, i.e. without looking at the data, and  that the groups in the same layer do not overlap. These restrictions do not pose a problem for our primary motivating application, although extensions to data-driven or overlapping groups are of interest as well---see the conclusion for a discussion. 

For each  layer $m$ of inference, then, we need a definition of null groups of variables. Extending the framework above, we define 
 group null hypotheses
\begin{equation}
H_g^m: Y \independent X_{\mathcal A_g^m} | X_{-\mathcal A_g^m},
\end{equation}
where $X_{\mathcal A_g^m} = \{X_j: j \in \mathcal A_g^m\}$ and $-\mathcal A_g^m = [N] \setminus \mathcal A_g^m$. Another natural definition might be based on the intersection hypotheses $\cap_{j \in \mathcal A_g^m} H_j$: while 
 in degenerate cases (e.g. when two variables are perfectly correlated), it might be the case that $H_g^m \neq \cap_{j \in \mathcal A_g^m} H_j$, this undesirable behavior often does not occur:
\begin{proposition} \label{prop:nondegeneracy}
Suppose that $(X_1, \dots, X_N) \in \mathcal D = \mathcal D_1 \times \cdots \times \mathcal D_N$, and that this joint distribution has nonzero probability density (or probability mass) at each element of $\mathcal D$. Then,
\begin{equation}
H_g^m = \bigcap_{j \in \mathcal A_g^m} H_j \quad \text{for all } g, m.
\label{equivalence}
\end{equation}
\end{proposition}

We assume for the rest of the paper that (\ref{equivalence}) holds, i.e. a group of variables is conditionally independent of the response if and only if each variable in that group is conditionally independent of the response.

The selected variables $\mathcal S$ induce group selections $\mathcal S_m \subset [G_m]$ at each layer via
\begin{equation*}
{\cal S}_m = \{g \in [G_m]: {\cal S} \cap {\cal A}_g^m \neq \varnothing\}:
\label{S_hat_m}
\end{equation*}
i.e. %we select
 a group is selected  if %we select 
 at least one variable belonging to that group is selected. To strive for replicable findings with respect to each layer of interpretation, we seek methods for which $\mathcal S_m$ has a low false discovery rate for each $m$. If $\mathcal H_0 = \{j \in [N]: H_j \text{ null}\}$ is the set of null variables, then the set of null groups at layer $m$ is
\begin{equation*}
\mathcal H_0^m = \{g \in [G_m]: \mathcal A_g^m \subset \mathcal H_0\}.
\end{equation*}
(this is guaranteed by (\ref{equivalence})). Then, the number of false discoveries at layer $m$ is
\begin{equation*}
V_m({\cal S}) = |{\cal S}_m \cap \mathcal H_0^m|,
\label{null_groups}
\end{equation*}
which we abbreviate with $V_m$ whenever possible without incurring confusion.
The corresponding false discovery rate is defined as the expectation of the false discovery proportion (FDP) at that layer: 
\begin{equation*}
\text{FDP}_m({\cal S}) = \frac{V_m(\cal S)}{|{\cal S}_m|} \quad \text{and}  \quad \text{FDR}_m = \mathbb E[\text{FDP}_m({\cal S})],
\label{group_FDR}
\end{equation*}
 using the convention that 0/0 = 0 (the FDP of no discoveries is 0). 
 %Suppose we have FDR tolerance levels $q_1, \dots, q_M$ for each layer. Then, we would like our 
 A selection procedure obeys \textit{multilayer FDR control} \citep{BR15} at levels $q_1, \dots, q_M$ for each of the layers if 
\begin{equation}
\text{FDR}_m \leq q_m \ \text{for all $m$.} \label{mlFDR}
\end{equation}
 
It might be surprising that the guarantee of %a procedure guaranteeing 
FDR control for the selection of  individual variables $X_1,\ldots, X_N$ does not %necessarily 
extend to the control $\text{FDR}_m$.
\begin{figure}[t!]
	\begin{center}
	\includegraphics[width = \textwidth]{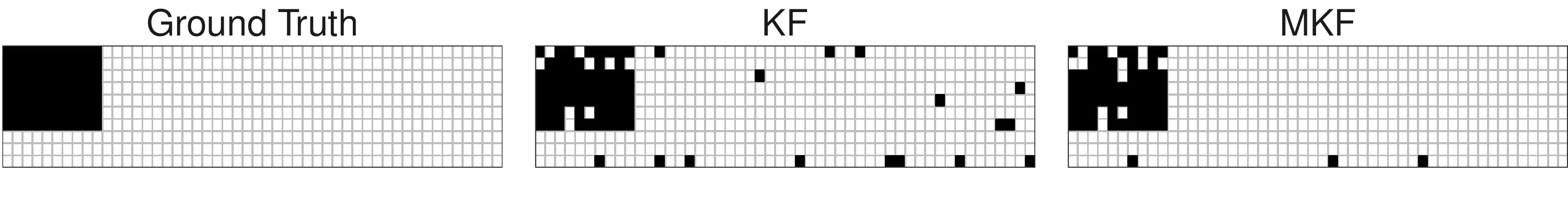}
	\end{center}
	\caption{Demonstration that small $\text{FDP}_{\text{ind}}$ does not guarantee small $\text{FDP}_{\text{grp}}.$ Each square represents a variable and columns contain variables in the same group. The left-most panel illustrates the true status of the hypotheses: a black square corresponds to non-null and a white square to a null variable. In the second and third panels, black squares represent  selected variables by KF and MKF respectively with $q_{\text{ind}}= q_{\text{grp}} = 0.2.$ The KF has $\text{FDP}_{\text{ind}} = 0.21$ and $\text{FDP}_{\text{grp}} = 0.58$, while MKF has $\text{FDP}_{\text{ind}} = 0.05$ and $\text{FDP}_{\text{grp}} = 0.17$.
	}
	\label{fig:mkn_advertisement}
\end{figure}
Figure \ref{fig:mkn_advertisement} provides an illustration of this fact for the simple case of $M=2$, with one layer  corresponding to the individual variables (denoted below by the subscript ``ind") and one group layer (denoted by the subscript ``grp").
We generate a matrix $\bm{X}\!$ ($n=1200$, $N=500$) % variables 
sampling each entry independently from a standard normal distribution. 
The 500 variables are organized % We organize variables 
into 50 groups of 10 elements each. %, each comprising 10 variables. 
The outcome $\bm{y}$ is generated from $\bm{X}\!$ using a linear model with 70 nonzero coefficients, %which we spread 
evenly spread across 10 groups. The middle panel of Figure \ref{fig:mkn_advertisement} shows the results of applying the knockoff filter \citep{BC15} with a target FDR level of 0.2. While the false discovery proportion for the individual layer is near the nominal level ($\text{FDP}_{\text{ind}} = 0.21$), the FDP at the group layer is unacceptably high ($\text{FDP}_{\text{grp}} = 0.58$). 
% The same phenomenon % observed in \cite{BR15}, which
  %motivated the development of the p-filter.  %We can let
  The middle panel of Figure \ref{fig:mkn_advertisement} guides our intuition for why this problem occurs: false  discoveries occur roughly uniformly across null variables and are then dispersed across groups (instead of being clustered in a small number of groups). When the number of null groups is comparable to or larger than the number of false discoveries, we have $V_{\text{grp}} \approx V_{\text{ind}}$ and  we can write roughly
\begin{equation*}
\text{FDP}_{\text{grp}} = \frac{V_{\text{grp}}}{|\mathcal S_{\text{grp}}|} \approx \frac{V_{\text{ind}}}{|\mathcal S_{\text{grp}}|} = \frac{|\mathcal S_{\text{ind}}|}{|\mathcal S_{\text{grp}}|}\frac{V_{\text{ind}}}{|\mathcal S_{\text{ind}}|} = \frac{|\mathcal S_{\text{ind}}|}{|\mathcal S_{\text{grp}}|}\text{FDP}_{\text{ind}}.
\end{equation*}
Hence, the group FDP is inflated roughly by a factor of $|\mathcal S_{\text{ind}}|/|\mathcal S_{\text{grp}}|$ compared to the individual FDP. This factor is high when we make several discoveries per group, as we expect %. In turn, this is likely
 when a non-null group has a high  number of non-null elements (high \textbf{saturation}).

To summarize, %as discussed in \cite{BR15},
 if we want to to make replicable discoveries at $M$ layers, %then 
we need to develop model selection approaches that explicitly target (\ref{mlFDR}). % This is precisely the goal of the
The multilayer knockoff filter precisely achieves this goal, as  is illustrated in the third panel of Figure \ref{fig:mkn_advertisement}, where much fewer variable groups (depicted as columns) have spurious discoveries compared to the results of the regular knockoff filter in the middle panel. 

\subsection{Background} \label{sec:background}

To achieve multilayer FDR control in the model selection setting, we capitalize on the properties of knockoff statistics \citep{BC15,CetL16} within a multilayer hypothesis testing paradigm described in the paper introducing the p-filter \citep{BR15}.
Here we briefly review these two methods.

\paragraph{Knockoff statistics: an alternative to p-values for variable selection} The knockoff filter \citep{BC15} is a powerful methodology for the variable selection problem that controls the FDR with respect to the single layer of individual variables. This method is based on constructing \textit{knockoff statistics} $W_j$ for each variable $j$. Knockoff statistics are an alternative to p-values; their distribution under the null hypothesis is not fully known, but obeys a sign symmetry property. In particular, knockoff statistics $\bm W = (W_1, \dots, W_N) \in \br^N$ for any set of hypotheses $H_1, \dots, H_N$ obey the \textit{sign-flip property} if conditional on $|\bm{W}|$, the signs of the $W_j$'s corresponding to null $H_j$ are distributed as i.i.d. fair coin flips. This property allows us to view $\text{sgn}(W_j)$ as ``one-bit p-values" and paves the way for FDR guarantees. Once $W_j$ are constructed, the knockoff filter proceeds similarly to the BH algorithm, rejecting $H_j$ for those $W_j$ passing a data-adaptive threshold.

% I do think you would be better off talking about knockoff scores here.
In the variable selection context, the paradigm for creating knockoff statistics $W_j$ is to create an artificial variable for each original to serve as a control. These \textit{knockoff variables} $\widetilde{\bm{X}}\in \br^{n \times N}$ are defined to be pairwise exchangeable with the originals, but are not related to $\bm{y}$. Then, $[\bm{X}\ \widetilde{\bm{X}}]$ are assessed jointly for empirical association with $\bm{y}$ (e.g. via a penalized regression), resulting in a set of \textit{variable importance statistics} $(\bm Z, \widetilde{\bm Z}) = (Z_1, \dots, Z_N, \widetilde Z_1, \dots, \widetilde Z_N)$. Each knockoff statistic $W_j$ is constructed as the difference (or any other antisymmetric function) of $Z_j$ and $\widetilde Z_j$. Hence, large positive $W_j$ provides evidence against $H_j$. Note that the sign of $W_j$ codes whether the original variable or its knockoff is more strongly associated with $\bm y$.  The sign-flip property is a statement that when $\bm{X}_j$ is not associated with $\bm{y}$, swapping columns $\bm X_j$ and $\widetilde{ \bm X}_j$ does not change the distribution of $(\bm Z, \widetilde{ \bm{ Z}})$. In its original form, the knockoff filter applied only to the fixed design low-dimensional linear model. Recently, \cite{CetL16} have extended it to the high-dimensional random design setting via \textit{model-X knockoffs} . Model-X knockoffs greatly expand the scope of the knockoff filter, and make it practical to apply for genetic association studies. Additionally, \textit{group knockoffs statistics} have been proposed by \cite{DB16} to test hypotheses of no association between groups of variables and a response. Since determining whether a group contains non-null variables is usually easier than pinpointing which specific variables are non-null, group knockoff statistics are more powerful than knockoff statistics based on individual variables. This is due mainly to increased flexibility in constructing $\widetilde{\bm X}$ (since only group-wise exchangeability is required), which allows $\bm X_j$ and $\widetilde{\bm{ X}}_j$ to be less correlated, which in turn increases power.

Rather than detailing here the methodology to construct knockoffs, variable importance statistics, and knockoff statistics, we defer these details to Section \ref{sec:multilayer_knockoff_filter}, where we describe our proposed method.

\paragraph{p-filter: a paradigm for multilayer hypothesis testing} The p-filter \citep{BR15, RetJ17} provides a framework for multiple testing with FDR control at multiple layers. This methodology, which generalizes the BH procedure, attaches a p-value to each group at each layer. Group p-values are defined from base p-values via the (weighted) Simes global test (though the authors recently learned  that this can be generalized to an extent, Ramdas (2017), personal communication). Unless specified otherwise, in this paper ``p-filter" refers to the methodology as described originally in \cite{BR15}. For a set of thresholds $\bm t = (t_1, \dots, t_M) \in [0, 1]^M$, base-level hypotheses are rejected if their corresponding groups at each layer pass their respective thresholds. A threshold vector $\bm t$ is ``acceptable" if an estimate of FDP for each layer is below the corresponding target level. A key observation is that the set of acceptable thresholds always has an ``upper right hand corner," which allows the data adaptive thresholds $\bm{t^*}$ to be chosen unambiguously as the most liberal acceptable threshold vector, generalizing the BH paradigm. 

\section{Multilayer Knockoff Filter} \label{sec:multilayer_knockoff_filter}

As illustrated in the next section, the  multilayer knockoff filter uses knockoff statistics---uniquely suited for variable selection---within the multilayer hypothesis testing paradigm of the p-filter. The p-filter paradigm was intended originally for use only with p-values, so justifying the use of knockoff statistics in this context requires a fundamentally new theoretical argument. Moreover, unlike the p-filter, statistics at different layers need not be ``coordinated" in any way, and indeed our theoretical result handles arbitrary between-layer dependencies. See Section \ref{sec:discussion} for more comparison of MKF with previous approaches.

% , but the multilayer FDR control of the resulting procedure is not an easy consequence of the FDR control results for either the knockoff filter or the p-filter. Moreover, while the p-filter \cite{BR15} methodology defines group p-values by aggregating p-values for base level hypotheses, MKF allows any set of group knockoff statistics at each layer. Indeed, we provide a theoretical FDR guarantee for MKF that assumes nothing about the between-layer dependencies of these statistics.

%In this section, we first present the proposed methodology, then state and outline the proof of our theoretical FDR control result, and close with a discussion.

\subsection{The procedure}
We first provide a high-level view of MKF in Framework \ref{framework:multilayer_knockoff_filter}, and then discuss each step in detail. 

\setcounter{algocf}{0}
\noindent
\begin{minipage}{\linewidth}
	\begin{algorithm}[H]
		\SetAlgorithmName{Framework}{}\; %last arg is the title of listing table		
		\KwData{$\bm{X}\!$, $\bm{y}$, partitions $\{{\cal A}_g^m\}_{g,m}$ with $g=1, \ldots, G_m$ and $m=1, \ldots, M$, FDR target levels $q_1, \dots, q_M$}
		\For{$m = 1$ \KwTo $M$}{
			Construct group knockoff variables $\widetilde{\bm{X}}^m$\;
			Construct group knockoff statistics $\bm{W}^m = (W_1^m, \dots, W_{G_m}^m) = w^m([\bm{X} \ \widetilde{\bm{X}}^m], \bm{y})$ satisfying the sign-flip property\;
		}	
		For $\bm{t} = (t_1, \dots, t_M) \in [0, \infty)^M$, define candidate selection set
		\begin{equation*}
		{\cal S}(\bm{t}) = \{j: W_{g(j, m)}^m \geq t_m \ \forall m\},
		\end{equation*}
		where $g(j,m)$ is the group at layer $m$ to which hypothesis $j$ belongs\;
		For each $m$, let $\widehat V_m(t_m)$ be an estimate of $V_m(\mathcal S(\bm t))$ and define $\widehat{\text{FDP}}_m(\bm{t}) = \displaystyle \frac{\widehat V_m(t_m)}{| {\cal S}_m(\bm{t})|}$\;
		Find $\bm{t}^* = \min\{\bm{t}: \widehat{\text{FDP}}_m(\bm t) \leq q_m \ \forall m\}$\;
		\KwResult{Selection set $ {\cal S} =  {\cal S}(\bm{t}^*)$.}
		\caption{\bf Multilayer Knockoff Filter}
		\label{framework:multilayer_knockoff_filter}
	\end{algorithm}
\end{minipage}

\paragraph{Constructing knockoffs for groups} 

To carry out our layer-specific inference, we need to construct knockoffs for groups of variables. 
%Group knockoffs can be constructed using either the fixed design linear model assumption or the model-X assumption, in parallel to \cite{BC15} and \cite{CetL16}, respectively. 

Within the fixed design framework, we can rely on the recipe  described in \cite{DB16}. These group knockoffs have the property %must be constructed so 
that the first two sample moments of the augmented design matrix $[\bm X\ {\widetilde{\bm X}}^m]$ are invariant when any set of groups is swapped with their knockoffs. To be more precise, let $\mathcal C_m \subset [G_m]$ be a set of groups at the $m$th layer, and let
\begin{equation}
\mathcal C = \bigcup_{g \in \mathcal C_m} \mathcal A_g^m 
\label{union_of_groups}
\end{equation}
be the variables belonging to any of these groups. Let $[\bm X\ {\widetilde{\bm X}}^m]_{\text{swap}(\mathcal C)}$ be the result of swapping the columns $\bm{X}_j$ and $\widetilde{\bm X}_j$ in the augmented design matrix for each $j \in \mathcal C$. Invariance of empirical first moments to swapping can be achieved by centering all variables and their knockoffs, while invariance of empirical second moments requires that for each $\mathcal C_m$, 
\begin{equation}
[\bm{X}\ \widetilde{\bm{X}}^m]_{\text{swap}({\cal C})}^T [\bm{X}\ \widetilde{\bm{X}}^m]_{\text{swap}({\cal C})} = [\bm{X}\ \widetilde{\bm{X}}^m]^T [\bm{X}\ \widetilde{\bm{X}}^m].
\label{covariance_swap_invariance}
\end{equation}
	
A construction of group knockoffs $\widetilde{\bm X}^m$ satisfying this property is given in \cite{DB16}. It is based on the observation that (\ref{covariance_swap_invariance}) is equivalent to
\begin{equation}
\begin{split}
&(\widetilde{\bm{X}}^m)^T \widetilde{\bm{X}}^m = \bm \Sigma = \bm X^T \bm X; \quad (\widetilde{\bm{X}}^m)^T \bm X = \bm \Sigma^m - \bm S^m; \\
& \qquad \qquad \bm S^m \succeq 0 \text{ group block diagonal}.
\label{equivalent}
\end{split}
\end{equation}
For a fixed block diagonal matrix $\bm S^m$, $\widetilde{\bm{X}}^m$ satisfying (\ref{equivalent}) can be constructed via
\begin{equation*}
\widetilde{\bm{X}}^m = \bm X(\bm I_{N} - \bm \Sigma^{-1}\bm S^m) + \bm{\widetilde U} \bm C^m,
\end{equation*}
where $\widetilde{\bm{U}}$ is an $n \times N$ orthonormal matrix orthogonal to the span of $\bm X$ and $\bm C^m$ is a Cholesky square root of $2\bm S^{m} - \bm S^m \bm \Sigma^{-1} \bm S^m$. The latter expression is positive semidefinite if $\bm S^m \preceq 2\bm \Sigma$. While there are several ways to construct a block diagonal $\bm S^m$ satisfying $0 \preceq \bm S^m \preceq 2\bm \Sigma$, the \textit{equicorrelated} knockoff construction is defined via
\begin{equation}
\bm{S}^m = \text{diag}(\bm S_1^m, \dots, \bm S_{G_m}^m), \quad \bm S_g^m = \gamma^m \cdot \bm \Sigma_{\mathcal A_g^m, \mathcal A_g^m},
\label{S_1}
\end{equation}
where 
\begin{equation*}
\gamma^m = \min(1, 2 \cdot \lambda_{\min}(\bm D^m \bm \Sigma \bm D^m))
\label{S_2}
\end{equation*}
and
\begin{equation}
\bm D^m = \text{diag}\left(\bm \Sigma_{\mathcal A_1^m, \mathcal A_1^m}^{-1/2}, \dots, \bm \Sigma_{\mathcal A_{G_m}^m, \mathcal A_{G_m}^m}^{-1/2}\right).
\label{S_3}
\end{equation}
Note that throughout this paper, all numerical experiments are based on the fixed design linear model and will thus use this construction.

Within the random design framework,
% group knockoffs have not yet been proposed, but 
the most natural way to define group knockoffs is to generalize the definition in \cite{CetL16}. For a set of random variables $(X_1, \dots, X_N)$ and a set of groups $\{\mathcal A_g^m\}_{g \in [G_m]}$, a set of \textit{model-X group knockoffs} $(\widetilde X_1^m, \dots, \widetilde X_N^m)$ is such that
\begin{equation}
(X_1, \dots, X_N, \widetilde X_1^m, \dots, \widetilde X_N^m)_{\text{swap}(\mathcal C)} \overset d = (X_1, \dots, X_N, \widetilde X_1^m, \dots, \widetilde X_N^m)
\label{dist_swap_invariance}
\end{equation}
for any $\mathcal C$ of the form (\ref{union_of_groups}). Note that for regular model-X knockoffs, (\ref{dist_swap_invariance}) must hold for \textit{all} $\mathcal C$. Hence, less exchangeability is required of model-X group knockoffs, which allows for them to be less similar to the original variables.

The sequential conditionally independent pairs (SCIP) procedure, proposed in \cite{CetL16} to prove the existence of model-X knockoffs, generalizes straightforwardly to model-X group knockoffs. In particular, the group SCIP procedure proceeds as follows: for each $g = 1, \dots, G_m$, we sample $\widetilde X^m_{\mathcal A_g^m}$ from $\mathcal L(X_{\mathcal A_g^m}|X_{-\mathcal A_g^m}, \widetilde X^m_{\mathcal A_1^m, \dots, \mathcal A_{g-1}^m})$. The more explicit second-order knockoffs construction (exact for normally distributed variables) also carries over fairly directly. Suppose that $X$ is distributed as $N(0, \bm \Sigma)$. Then, sampling 
\begin{equation*}
\widetilde X^m|X \sim N(\bm \mu, \bm V), \quad \bm \mu = X - X \bm \Sigma^{-1}\bm S^m, \quad \bm V = 2 \bm S^m - \bm S^m \bm \Sigma^{-1} \bm S^m
\end{equation*}
yields a valid group model-X knockoff construction, where $\bm S^m$  is as defined in (\ref{S_1}), (\ref{S_2}) and (\ref{S_3}). Note that $X$ and $\bm \mu$ are treated as row vectors of dimension $1 \times N$. Recently, an HMM-based model-X knockoff construction has been proposed by \cite{SetC17}, tailored specifically for genetic design matrices. This construction can be generalized to the group setting as well (Matteo Sesia, personal communication, 2018), though we do not discuss the details in the present work.
\paragraph{Constructing group importance statistics} 

Once group knockoffs are constructed for each layer $m$, the next step is to define group importance statistics $(\bm Z^m, \widetilde{\bm Z}^m) = (Z^m_1, \dots, Z_{G_m}^m, \widetilde Z_1^m, \dots, \widetilde Z_{G_m}^m) = z^m([\bm{X}\ \widetilde{\bm{X}}^m], \bm y)$, one for each group and each knockoff group. Here $z^m$ is a function that assesses the association between each group (original and knockoff) and the response. 

The function $z^m$ must be \textit{group-swap equivariant}, i.e. 
\begin{equation*}
z^m([\bm{X}\ \widetilde{\bm{X}}^m]_{\text{swap}({\cal C})}, y) = (Z^m_1, \dots, Z_{G_m}^m, \widetilde Z_1^m, \dots, \widetilde Z_{G_m}^m)_{\text{swap}({\cal C}_m)},
\end{equation*} 
for each $\mathcal C_m \subset [G_m]$ and corresponding $\mathcal C$ defined in (\ref{union_of_groups}). In words, this means that swapping entire groups in $[\bm X \ \widetilde{\bm{ X}}^m]$ translates to swapping the entries of $(\bm Z^m, \widetilde{\bm Z}^m)$ corresponding to those groups. For fixed design knockoffs, $z^m$ must also satisfy the \textit{sufficiency property}: $z^m$ must operate on the data only through the sufficient statistics $[\bm{X}\ \widetilde{\bm{X}}^m]^T \bm{y}$ and $[\bm{X}\ \widetilde{\bm{X}}^m]^T [\bm{X} \ \widetilde{\bm{X}}^m]$. Taken together, these steps lead to a function $w^m$ mapping the augmented design matrix and response vector to a vector of knockoff statistics $\bm W^m = (W_1^m, \dots, W_{G_m}^m)$.

There are many possible choices of $z^m$, and any choice satisfying group-swap equivariance (and sufficiency, for fixed design knockoffs) is valid. However, different choices will lead to procedures with different power. Generalizing the proposal of \cite{DB16}---based on the group lasso---we consider a class of group importance statistics 
$(\bm Z^m, \widetilde{\bm Z}^m)$ obtained by first solving for each $\lambda$ the penalized regression
\begin{equation}
\begin{split}
\bm b^\star(\lambda), \widetilde{\bm b^\star}(\lambda) = &\underset{\bm{b}, \widetilde{\bm{b}}}{\arg \min}\ \frac{1}{2}\norm{\bm{y} - [\bm{X}\ \widetilde{\bm{X}}^m]{\bm{b} \choose \widetilde{\bm{b}}}}^2 \\
&\qquad \qquad + \lambda\left(\sum_{g = 1}^{G_m} \ell_g^m(\bm{b}_{{\cal A}_g^m}) + \sum_{g = 1}^{G_m} \ell_g^m(\widetilde{\bm{b}}_{{\cal A}_g^m})\right),
\end{split}
\label{augmented_group_lasso_general}
\end{equation}
where $\ell_g^m$ are arbitrary penalty functions (the group lasso corresponds to $\ell_g^m(\bm{u}) = \sqrt{|{\cal A}_g^m|}\norm{\bm{u}}_2$). The group importance statistics $Z_g^m$ ($\widetilde Z_g^m$) are then defined as the value of $\lambda$  for which the group $\bm{X}_{\mathcal A_g^m}$ ($\widetilde{\bm X}_{\mathcal A_g^m}$) first enters the lasso path:
\begin{equation}
\begin{split}
Z_g^m = \sup\{\lambda: \bm{b}^{\star}_{\mathcal A_g^m}(\lambda) \neq 0\}, \quad \widetilde Z_g^m = \sup\{\lambda: \widetilde{\bm{b}}^{\star}_{\mathcal A_g^m}(\lambda) \neq 0\}.
\end{split}
\label{first_entry}
\end{equation}
For each $m$, the regularization in  (\ref{augmented_group_lasso_general}) is defined on subsets of $\bm{b}$ corresponding to groups at the $m$ layers and it is separable with respect to the $m$th partition $\{{\cal A}_g^m\}$. While this guarantees  group-swap equivariance and sufficiency, other constructions are certainly possible.
%Note that we require the regularization penalty term in (\ref{augmented_group_lasso_general}) to be separable with respect to the $m$th partition $\{{\cal A}_g^m\}$. This guarantees  group-swap invariance. Another way 
%This separability is important in ensuring that $z_m(\cdot)$ is group swap equivariant. 	It is clear that not all regularization penalties will lead to this equivariance property, and in particular, we cannot naively incorporate the group structure from all layers into the definition of $(\bm{Z}^m, \widetilde{\bm{Z}}^m)$. Nevertheless, it is entirely possible to carefully construct penalties other than that in (\ref{augmented_group_lasso_general}), potentially incorporating structure from multiple layers, to satisfy group swap equivariance.

Different choices of  $\ell_g^m$ allow us to adapt to the available information on signal structure and potentially gain power.
%The reason for allowing ourselves freedom in choosing the group penalties $\ell_g^m$ is to gain power.
The 2-norm defining the group lasso leads to entire groups entering the regularization path at the same time, so each variable in a group contributes equally to the corresponding group importance statistic. Drawing an analogy to global testing, this definition is similar to the Fisher combination test or the chi-squared test, which are known to be powerful in regimes when the signal is weak and distributed. In our case, we should apply group lasso based statistics if we suspect that each group has many non-nulls. Taking this analogy further, we can also construct a Simes-like statistic that is more suited to the case when we believe each group has a few strong signals. This test statistic is defined by letting $\ell_g^m(\bm{u}) = \norm{\bm{u}}_1$; i.e. running a regular lasso regression. This will allow for each variable to come in on its own, and the knockoff statistic $W_g^m$ will be driven by the strongest predictor in the corresponding group. 

\paragraph{Constructing group knockoff statistics}

Finally, group knockoff statistics are defined via $W_g^m = f_g^m(Z_g^m, \widetilde Z_g^m)$, where $f_g^m$ is any antisymmetric function (i.e. swapping its arguments negates the output), such as the difference $f(Z_g^m, \widetilde Z_g^m) = Z_g^m - \widetilde Z_g^m$ or the signed-max $f(Z_g^m, \widetilde Z_g^m) = \max(Z_g^m, \widetilde Z_g^m)\cdot \text{sgn}(Z_g^m - \widetilde Z_g^m)$. Hence, $W_g^m$ quantifies the difference in how significantly associated $\bm X_{\mathcal A_g^m}$ and $\widetilde{\bm X}_{\mathcal A_g^m}$ are with $\bm y$.

\paragraph{Definition of $\widehat{\bm{V}}_m(t_m)$} The FDR guarantee for the method depends crucially on the estimate $\widehat V_m(t_m)$ of $V_m(\mathcal S(t))$. Intuitively, the larger $\widehat V_m(t_m)$ is, the stronger the FDR guarantee. For the original knockoff filter \citep{BC15}, two choices of $\widehat V(t)$ were considered: $\widehat V(t) = |\{j: W_j \leq -t\}|$ (which leads to a procedure we abbreviate KF, for knockoff filter) and $\widehat V(t) = 1 + |\{j: W_j \leq -t\}|$ (leading to a procedure we abbreviate KF+). Note that these are defined in terms of the size of the set obtained from $\mathcal S(t) = \{j: W_j \geq t\}$ by reflection about the origin. These definitions are motivated by the sign-flip property, and it can be easily shown that $\widehat V(t) = |\{j: W_j \leq -t\}|$ is a conservative estimate for $V(\mathcal S(t)) = |\{\text{null } j: W_j \geq t\}|$. The reason for also considering $\widehat V(t) = 1 + |\{j: W_j \leq -t\}|$ is that the extra 1 is needed for exact FDR control. The KF procedure, defined without this extra 1, controls a weaker criterion, called the mFDR, and defined as follows:
\begin{equation*}
\text{mFDR} = \mathbb E\left[\frac{| {\cal S} \cap \mathcal H_0|}{| {\cal S}| + q^{-1}}\right].
\end{equation*}

Similarly, we consider methods MKF and MKF+ based on two definitions of $\widehat V_m(t_m)$. However, we shall also consider the effect of adding a constant multiplier $c$ to this estimate as well; see Procedures \ref{alg:MKF} and \ref{alg:MKF+}.

\setcounter{algocf}{0}
\noindent
\begin{minipage}{\linewidth}
	\begin{algorithm}[H]
		\SetAlgorithmName{Procedure}{}\; %last arg is the title of listing table
		Framework \ref{framework:multilayer_knockoff_filter}, with $\widehat V_m(t_m) = c\cdot |\{g: W_g^m \leq -t_m\}|$.
		\caption{\bf MKF($\bm{c}$)}
		\label{alg:MKF}
	\end{algorithm}
\end{minipage}
\begin{minipage}{\linewidth}
	\begin{algorithm}[H]
		\SetAlgorithmName{Procedure}{}\; %last arg is the title of listing table
		Framework \ref{framework:multilayer_knockoff_filter}, with $\widehat V_m(t_m) = c\cdot (1 + |\{g: W_g^m \leq -t_m\}|)$.
		\caption{\bf MKF$\bm{(c)+}$}
		\label{alg:MKF+}
	\end{algorithm}
\end{minipage}

\paragraph{Definition and computation of $\bm{t}^*$}

Note that the last step in Framework \ref{framework:multilayer_knockoff_filter} needs clarification, since the minimum of a set in $M$ dimensions is not well defined in general. However, the following lemma resolves the issue.
\begin{lemma}
Consider the set of valid thresholds
\begin{equation}
\mathcal T = {\mathcal T}(q_1, \dots, q_M) = \{\bm{t}: \widehat{\text{FDP}}_m(\bm{t}) \leq q_m \text{ for all } m\}.
\label{valid_thresholds}
\end{equation} 
For any definition of $\widehat V_m(\cdot)$ in Framework \ref{framework:multilayer_knockoff_filter} depending only on $t_m$ (as opposed to the entire vector $\bm t$), the set ${\mathcal T}$ will possess the ``lower left-hand corner property," which means that it contains the point $\bm t^* = (t_1^*, \dots, t_M^*)$ defined by 
\begin{equation*}
t^*_m = \min\{t_m: (t_1, \dots, t_M) \in {\mathcal T} \text{ for some } t_1, \dots, t_{m-1}, t_{m+1}, \dots, t_M\}.
\end{equation*}
\end{lemma}
Hence, the point $\bm t^*$ is the lower left-hand corner of $\mathcal T$ and is the minimum in the last step of Procedure \ref{framework:multilayer_knockoff_filter}. The p-filter enjoys this same property, and in fact  the proof of this lemma is the same as that of Theorem 3 in \cite{BR15}. In addition to being well-defined, the threshold $\bm{t}^*$ can be computed efficiently using the same iterative coordinate-wise search algorithm proposed in \cite{BR15}.

Figure \ref{fig:mkn_illustration} provides an  illustration of the multilayer knockoff filter on  simulated data.  The 2000 variables are broken into 200 groups, each of size 10.  (More details about this example %simulation
are %provided
in Section \ref{sec:mkn_simulations}; it corresponds to the high saturation setting, with $\text{SNR} = 0.5$ and variable correlation 0.1.) The multilayer knockoff filter for the individual layer and the group layer results in the selection region in the upper right-hand corner and enjoys high power. 
%Each point in Figure \ref{fig:mkn_illustration} represents an individual-layer hypothesis $j$, with coordinates  $W_j^1$ and $W_{g(j,2)}^2$.  In this simulation setting (with relatively high signal-to-noise ratio and relatively low feature correlation), the non-nulls cluster in the upper right-hand corner, leading to high power. The transparent blue lines illustrate the regions of the space we use to define $\widehat V_m(t_m)$. For example, the number of points in the shaded half-plane to the left of the vertical transparent blue line is $\widehat V_1 (t_1^*)$.
For comparison, the (single layer) knockoff filter selects all points to the right of the broken %cyan
 line; among these there are several nulls not selected by MKF as their group signal is not strong enough. %The multilayer knockoff filter does not select several null points that are selected by the knockoff filter because their groups are not significant enough.
In this simulation, MKF reduces  false positives without losing any power.  
\begin{figure}[h!]
	\centering
	\includegraphics[width = 0.9\textwidth]{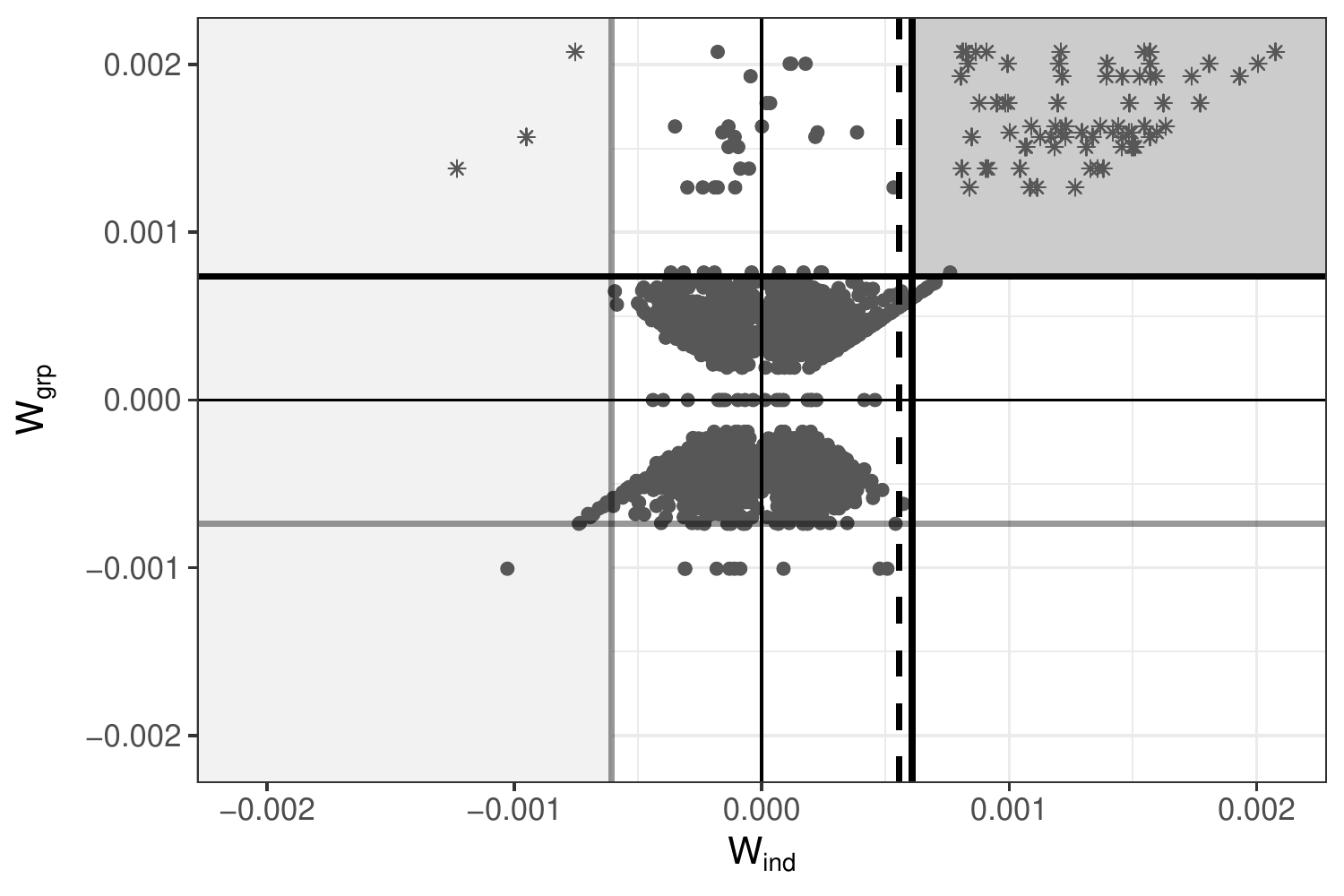}
	\caption{Illustration of one run of the multilayer knockoff filter. Each point  represents an individual-layer hypothesis $j$, with coordinates  $W_j^1$ and $W_{g(j,2)}^2$: circles indicate nulls, asterisks  non-nulls. The solid black lines are the thresholds for multilayer knockoff filter, while reflected gray lines are used in the definition of $\widehat V_m$. The broken black line represents the threshold for knockoff filter. The darkly shaded upper right corner represents the selection set of the multilayer knockoff filter, and the lightly shaded left half-plane represents  the area used to define $\widehat V_{\text{ind}}$.}
	\label{fig:mkn_illustration}
\end{figure}

\paragraph{Computational complexity}
Computationally, both the multilayer knockoff filter and the original knockoff filter can be broken down into three distinct parts: constructing knockoff variables, computing knockoff statistics, and then filtering those statistics to obtain a final selection set. Usually, the bottleneck is the second step; computing knockoff statistics involves solving a regularized regression of size $n \times 2N$, which costs $O(nN^2)$, assuming for the moment $n > 2N$. Constructing knockoff variables, which for the genetic application at hand is best done using HMM knockoffs \citep{SetC17}, only costs $O(nN)$, and the filtering costs $O(N)$. The same computational costs apply to the multilayer knockoff filter, except the first two steps must be done $M$ times. In summary, the computational cost of the multilayer knockoff filter is equivalent to that of solving $M$ regularized regressions. Solving regularized regressions is a standard computational task, and optimized solvers exist for these purposes such as \texttt{glmnet} \citep{friedman2010regularization} or \texttt{SparSNP} \citep{abraham2012sparsnp}, the former being general-purpose and the latter being specialized for genetics applications. Indeed, the feasibility of knockoff analysis of GWAS data has been already illustrated in \cite{CetL16} and \cite{SetC17}. Moreover, computations of the knockoff statistics for the $M$ layers are independent and thus can be run in parallel, ensuring that the multilayer knockoff filter scales well to genome-wide data sets.

\subsection{Theoretical guarantees}

Our main theoretical result guarantees that the MKF$(c_{\text{kn}})+$ has multilayer FDR control at the target levels, and that MKF$(c_{\text{kn}})$ has multilayer mFDR control at the target levels, where $c_{\text{kn}} = 1.93$. 
\begin{theorem} \label{kn_multilayer_fdr_control}
	Suppose $\bm{W}^m$ obeys the sign-flip property for all $m$. Then, the MKF($c$)+ method satisfies
	\begin{equation*}
	\text{\em FDR}_m \leq \frac{c_{\text{\em kn}}}{c} q_m\quad \text{ for all } m,
	\end{equation*} 
	where $c_{\text{\em kn}} = 1.93$. The MKF($c$) method satisfies
	\begin{equation*}
	\text{\em mFDR}_m \leq \frac{c_{\text{\em kn}}}{c}q_m \quad \text{ for all } m.
	\end{equation*} 
	In particular, the MKF($c_{\text{\em kn}}$)+ and MKF($c_{\text{\em kn}}$) methods have multilayer FDR control and multilayer mFDR control, respectively, at the target levels.	
\end{theorem}

While deferring technical lemmas to the supplementary materials, we outline here the essential steps of the proof as they differ fundamentally from those of KF and p-filter.
%This result requires a novel proof technique because it differs fundamentally from the corresponding results for both the knockoff filter and the p-filter.
 The proof of FDR control for the knockoff filter relies on a clever martingale argument that %. This argument 
 depends heavily on the fact that the threshold $t$ is one-dimensional; %for the knockoff filter, so
  the cutoff $t^*$ can be viewed as a stopping time with respect to a certain stochastic process.
   Instead,  we are dealing with an $M$-dimensional threshold $\bm{t}^*$ whose entries depend on the values of 
   $\bm{W}^m$ for all $m$. As the knockoff statistics  have complex dependencies with each other, we cannot represent $t_m$ as a stopping time with respect to a process that depends only on $\bm{W}^m\!$. 
   The p-filter being a multilayer method,  the proof of  FDR control deals with the complex nature of the threshold $\bm{t}^*\!$. However, by defining the  p-values at each layer from the individual hypotheses  p-values  with a set rule, \cite{BR15} 
   have a good handle on the dependencies between $p_g^m$ across layers and use this crucially in the proof.
   % These specific dependencies are used crucially to obtain the multilayer FDR control result. In our case, 
   In contrast, we intentionally avoid specifying the relations %having nice dependency structures 
between $\bm{W}^m$ for different $m$.%, so we cannot hope to use the same proof technique.
%Next, we outline the essential steps of the proof, while saving technical lemmas for the appendix. 

\begin{proof}
We prove FDR control for MKF($c$)+; the result for MKF($c$) follows from a very similar argument. We start introducing the following quantities: %first make the following definitions:
\begin{equation*}
 V_m^{+}(t_m) = |\{g: W_g^m \geq t_m\} \cap \mathcal H_0^m|, \quad  V_m^{-}(t_m) = |\{g: W_g^m \leq -t_m\} \cap \mathcal H_0^m|.
\end{equation*}
Note that both $V_m^+(t_m)$ and  $V_m^{-}(t_m)$ are defined in terms of the $m$th layer only and  that $V_m^{+}(t_m) = V_m(\mathcal S(0, \dots, 0, t_m, 0, \dots, 0))$, while $ V_m^-(t_m)$ is similar to $\widehat V_m(t_m)$. It is easy to verify that these two quantities satisfy
\begin{equation*}
 V_m^+(t_m) \geq V_m(\mathcal S(\bm{t})), \quad \widehat V_m(t_m) \geq c(1 +  V_m^-(t_m)). 
\end{equation*}
Then, for each $m$, we have
\begin{equation}
\begin{split}
\text{FDR}_m = \mathbb E[\text{FDP}_m(\bm{t}^*)] &= \mathbb E\left[\frac{V_m(\mathcal S(\bm{t}^*))}{ |{\cal S}_m(\bm{t}^*)|}\right] \\
&= \mathbb E\left[\frac{V_m(\mathcal S(\bm{t}^*))}{ |{\cal S}_m(\bm{t}^*)|}I(\bm{t}^* < \bm \infty)\right] \\
&= \mathbb E\left[\frac{V_m(\mathcal S(\bm{t}^*))}{\widehat V_m(t^*_m)}\frac{\widehat V_m(t^*_m)}{|{\cal S}_m(\bm{t}^*)|}I(\bm{t}^* < \bm \infty)\right] \\
&\leq q_m \cdot \frac{1}{c}\mathbb E\left[\frac{ V_m^+(t^*_m)}{1 +  V_m^-(t^*_m)}\right] \\
&\leq q_m \cdot \frac{1}{c}\mathbb E\left[\sup_{t_m}\frac{ V_m^+(t_m)}{1 +  V_m^-(t_m)}\right].
\label{sufficient_cond_proof}
\end{split}
\end{equation}
Hence it suffices to show that
\begin{equation}
\mathbb E\left[\sup_{t_m \geq 0}\frac{ V_m^+(t_m)}{1 +  V_m^-(t_m)}\right] \leq c_{\text{kn}}.
\label{bounded_ratio_multilayer}
\end{equation}
The introduction of the supremum over $t_m$ in the last equation is a key step in the proof: it makes the random variables in the expectation (\ref{bounded_ratio_multilayer}) depend only on the knockoff statistics at the $m$th layer,  decoupling the problem across layers and allowing any type of dependence between statistics for different values of $m$.%, which allows us not to assume anything about between-layer dependencies of feature statistics.

Given that we are working with quantities defined in one layer only,
%this decoupling,
 we  can %might as well
  drop the subscript $m$, and consider (\ref{bounded_ratio_multilayer}) as a statement about any set of knockoff statistics $(W_1, \dots, W_G)$ satisfying the sign-flip property. Hence, $W_g^m, V_m^+(t_m), V_m^-(t_m)$ become $W_g, V^+(t), V^-(t)$, respectively, and so on.
We are left with
\begin{equation}
\begin{split}
\mathbb E\left[\sup_{t \geq 0}\frac{V^+(t)}{1 +  V^{-}(t)}\right] &= \mathbb E\left[\sup_{t \geq 0}\frac{|\{g: W_g \geq t\} \cap \mathcal H_0|}{1 + |\{g: W_g \leq -t\} \cap \mathcal H_0|}\right] \\
&=\mathbb E\left[\mathbb E\left[\left.\sup_{t \geq 0}\frac{|\{g: W_g \geq t\} \cap \mathcal H_0|}{1 + |\{g: W_g \leq -t\} \cap \mathcal H_0|}\right| |\bm{W}|\right]\right].
\label{step_one}
\end{split}
\end{equation}

Now, consider ordering $\{W_g\}_{g \in \mathcal H_0}$ by magnitude: $|W_{(1)}| \geq \cdots \geq |W_{(G_0)}|$, where $G_0 = |\mathcal H_0|$. Let $\sigma_g = \text{sgn}(W_{(g)})$. By the sign-flip property, $\sigma_g$ are distributed as i.i.d. coin flips independently of $|\bm W|$. Moreover, note that the quantity inside the expectation is constant for all $t$ except $t \in \{|W_{(1)}|, \dots, |W_{(G)}|\}$, and for $t = |W_{(k)}|$, we have
\begin{equation*}
\begin{split}
\frac{|\{g: W_g \geq t\} \cap \mathcal H_0|}{1 + |\{g: W_g \leq -t\} \cap \mathcal H_0|} &=\frac{|\{g: W_g \geq |W_{(k)}|\} \cap \mathcal H_0|}{1 + |\{g: W_g \leq -|W_{(k)}|\} \cap \mathcal H_0|} \\
&=  \frac{|\{g \leq k: \sigma_g = +1\}|}{1 + |\{g \leq k: \sigma_g = -1\}|}.
\end{split}
\end{equation*}
Putting these pieces together, we have
\begin{equation*}
\begin{split}
\mathbb E\left[\sup_{t \geq 0}\frac{ V^+(t)}{1 +  V^{-}(t)}\right] &= \mathbb E\left[\max_{k \leq G_0}\frac{|\{g \leq k: \sigma_g = +1\}|}{1 + |\{g \leq k: \sigma_g = -1\}|}\right].
\label{step_two}
\end{split}
\end{equation*}
We can think of $\sigma_g$ as the increments of a simple symmetric random walk on $\mathbb Z$. The numerator above represents the number of steps to the right this walk takes, and the denominator the number of steps to the left. The quantity we are bounding is essentially the maximum over all steps in the walk of the ratio of steps right to steps left, averaged over all realizations of the random walk. Let $S_{k} = %\sum_{g = 1}^{k_m} X_g = 
|\{g \leq k: \sigma_g = +1\}|$ be the number of steps right and $k - S_{k} = |\{g \leq k: \sigma_g = -1\}|$ the number of steps left. It suffices to show that
\begin{equation}
\mathbb E\left[\sup_{k\geq 0}\frac{S_k}{1 + k - S_k}\right] \leq 1.93.
\label{sufficient}
\end{equation}
This is the content of Lemma \ref{bound_knockoffs}, which is proved in the supplementary materials.
\end{proof}

\subsection{Relations with other methods} \label{sec:discussion}

\subsubsection*{Comparison to other structured testing methods}

When scientific hypotheses have a complex structure, even formulating inferential guarantees is nontrivial; the multilayer hypothesis testing approach proposed in \cite{BR15} and used in our work is one of several options. Approaches to testing hypotheses at multiple levels vary in two key features: the way the space of hypotheses is traversed and the way families to be tested are defined. We illustrate the different approaches using the two-layer setup considered in the introduction. If the hypotheses have a nested (i.e. tree) structure, then it is common to traverse them hierarchically: one starts by testing groups and then proceeds to test individual hypotheses within rejected groups. The procedures described in \cite{Y08,BB14,BetS17b} follow this hierarchical approach. An alternative to hierarchical hypothesis traversal is to consider the multiple testing problem from the point of view of individual-level hypotheses: rejecting a set of individual-level hypotheses induces the rejections of the groups that contain them at each layer of interest. This is the approach taken by the p-filter \citep{BR15}. By testing hypotheses only if their corresponding groups were rejected, hierarchical approaches have the advantage of a smaller multiplicity burden. On the other hand, defining selections at each layer via the individual-level hypotheses has the advantage that it applies equally well to non-hierarchical ways of grouping hypotheses. The second dichotomy in based on how one defines families to be tested: either each group is a family of its own, or each resolution is a family of its own. For instance, the former corresponds to SNPs being tested against other SNPs in the same gene, and the latter corresponds to testing all SNPs against each other as one family. The methods of \cite{BB14} and \cite{BetS17b} take the former approach, while those of \cite{Y08} and \cite{BR15} take the latter. Both choices can be meaningful, depending on the application. In this work, we define discoveries using individual-level hypotheses as this marries well with the multiple regression framework and does not limit us to nested groups.
% because this reflects what is usually done in association studies and also does not preclude non-hierarchical ways of grouping hypotheses. 
We  treat each resolution (instead of each group) as a family because discoveries are often reported by type (e.g. as a list of SNPs or a list of genes), so FDR guarantees for each type are appropriate. These two choices align our testing framework with that of the p-filter.

\subsubsection*{The multiplier $c_{\text{kn}}$: its origins and impact} \label{sec:comparison_pfilter}
In addition to  using knockoff statistics instead of p-values, the multilayer knockoff filter differs from the p-filter \citep{BR15} in that it does not start from a set of individual-level statistics and construct group-level ones using specific functions of these:  instead the statistics $\bm{W}^m$ are constructed starting directly from the original data. This decision involves a trade-off: we get a more general procedure (and theoretical result) at the cost of a looser bound. 

By making no assumptions on the between-layer dependencies of $W_g^m\!$, the multilayer knockoff filter allows extra flexibility that can translate into greater power. For example, there might be different sources of prior information at the SNP and the gene levels: the analyst can use each source of information at its respective layer to define more powerful knockoff statistics (based on specific penalties)  without worrying about coordinating these statistics in any way. Even if  the same penalization is used in all layers, there is a potential power increase due to the fact that we can use group knockoff variables rather than individual ones. This advantage is especially pronounced if none of the layers consists of singletons.

The price we pay for this generality is the multiplier  $c_{\text{kn}} = 1.93$ in Theorem 1. To understand its effect,  note that in Procedures 4 and 5, by analogy with KF,  the natural choice is $c=1$ and define $\text{MKF} = \text{MKF}(1)$ and $\text{MKF}+ = \text{MKF}(1)+$.
Then Theorem 1 states that MKF+ (MKF) has an FDR (mFDR)  that is bounded by $c_{\text{kn}}q_m$.
%On the other hand, we get this generality at the cost of the multiplier $c_{\text{kn}} = 1.93$ in the definition of $\widehat V_m(t_m)$. Another way to view this is that the natural way to formulate the procedure, taking $c = 1$, leads to an FDR (mFDR) bound of $c_{\text{kn}}q_m$. 
Compare this to the theoretical result for the p-filter, which is shown to have exact multilayer FDR control: by explicitly leveraging the joint distribution of $p_g^m$, \cite{BR15} get a handle on the complicated thresholds $t^*_m$ and get a tight result. Meanwhile, our constant multiplier comes from the introduction of the supremum in (\ref{sufficient_cond_proof}): this amounts to a worst-case analysis, which for most constructions of $W_g^m$ will not be tight. 

Indeed, across all our simulations in Section \ref{sec:numerical}, we find that MKF+ has multilayer FDR control at the target levels (i.e. the constant is not necessary). Hence, we recommend that practitioners apply the MKF or MKF+ methods, without worrying about the correction constants. We view our theoretical result as an assurance that even in the worst case, the FDRs of MKF at each layer will not be much higher than their nominal levels $q_m$.

\subsubsection*{Generalized p-filter}

On the heels of the above discussion, we  
%The preceding discussion suggests that the p-filter itself can be generalized by allowing p-values to be constructed separately for each layer. Let us
 define the \textit{generalized p-filter}, a procedure that is the same as the p-filter, except that the p-values $p_g^m$ are any valid p-values for the hypotheses in layer $m$.
% can be defined in any way. In fact, we can get a multilayer FDR control guarantee for this method as well.
\begin{theorem} \label{th:pf_theorem}
	Suppose for each $m$, the null p-values among $\{p_g^m\}$ are independent and uniformly distributed. Then, the generalized p-filter satisfies
	\begin{equation*}
	\text{\em mFDR}_m \leq c_{\text{\em pf}}(G_m) \cdot q_m \quad \text{ for all } m,
	\end{equation*} 
	where $c_{\text{\em pf}}(G) = 1 + \exp\left(G^{-1/2} + \frac12 G^{-1}\right)0.42 + eG^{-1/4}$.
\end{theorem}
\begin{remark}
	Unlike for the multilayer knockoff filter, note that we do not have one universal constant multiplier $c_{\text{\em pf}}$. Instead, we get a bound $c_{\text{\em pf}}(G_m)$ that depends on the number of groups at each layer. However, strong numerical evidence suggests that in fact we can replace $c_{\text{\em pf}}(G_m)$ in the theorem with its limiting value 1.42. See Remark \ref{monotonicity} in the supplementary materials for additional comments. Moreover, the assumption of independent null p-values can potentially be relaxed to a PRDS assumption, but we have not explored this avenue. 
\end{remark}
\begin{proof}
	By similar logic as in the proof of Theorem \ref{kn_multilayer_fdr_control}, it suffices to verify the sufficient condition
	\begin{equation*}
	\mathbb E\left[\sup_{t_m \in [0,1 ]}\frac{|\{g \in \mathcal H_0^m: p_g^m \leq t_m\}|}{1 + G_m t_m}\right]  \leq c_{\text{pf}}(G_m). 
	\end{equation*}
	Again, note that the problem decouples across layers and we may drop the subscript $m$. Now, let $p_1, \dots, p_G$ be a sequence of i.i.d. uniform random variables, and let $F_G(t)$ be their empirical CDF. Then, it suffices to show that 
	\begin{equation*}
	\mathbb E\left[\sup_{t \in [0, 1]}\frac{F_G(t)}{G^{-1} + t}\right] \leq  1 + \exp\left(G^{-1/2} + \frac12 G^{-1}\right)0.42 + eG^{-1/4}.
	\end{equation*}
	This is the content of Lemma \ref{bound_pfilter}, which is proved in the supplementary materials.
\end{proof}

Recently \citep{RetJ17}, the p-filter methodology has been generalized to allow for more general constructions of group p-values, but must use \textit{reshaping} (a generalization of the correction proposed by \cite{BY01}) to guarantee FDR control. Hence, both generalizations must pay for arbitrary between-layer dependencies; our method with the constant $c_{\text{pf}}$ and the p-filter with reshaping.

\subsubsection*{Power of multilayer methods} \label{sec:power_fdr}

By construction, the multilayer algorithms we propose are at most as powerful as their single-layer versions. For our purposes, groups of variables function as inferential units, and not as prior information used to boost power (e.g. as in \cite{LB16}), although there is no reason groups cannot serve both functions within our framework. So while our methods are designed  to provide more FDR guarantees, it is relevant to evaluate the cost in terms of power of these additional guarantees.

%For example, the multilayer knockoff filter cannot be more powerful than the single-layer knockoff filter (assuming the same knockoff statistics and FDR thresholds are used). 

Consider controlling FDR for individual variables and for groups, compared to just controlling FDR for individual variables. When adding a group FDR guarantee, power loss depends on the group signal strength, the power of group statistics, and the desired group FDR control level $q_{\text{grp}}$. Power will decrease to the extent that the signal strength at the group layer is weaker than the signal strength at the individual layer. Assuming for simplicity that non-null variables have comparable effect sizes, group signal is weak when saturation is low (recall from the introduction that saturation is the average number of non-null variables per non-null group). Also, if the sizes of the groups vary, then group signal will be weaker if the non-null hypotheses are buried inside very large groups. Even if group signal is not too weak, the power of multilayer procedures will depend on the way group statistics are chosen. In particular, power will be better if Simes (or Simes-like) statistics are used if groups have a small number of strong signals, and if Fisher (or Fisher-like) statistics are used in the case of weak distributed effects. Finally, it is clear that lowering $q_{\text{grp}}$ will lower power.

As a final note, all of the multilayer methods discussed so far have a feature that might unnecessarily limit their power. This feature is the definition of $\widehat V_m = \widehat V_m(t_m)$ in terms of only the threshold $t_m$. Since the selection set $ {\cal S}(\bm{t})$ is defined with respect to the $M$-dimensional vector of thresholds $\bm{t}$, a definition of $\widehat V_m$ depending on this entire vector would be a better estimate of $V_m(\mathcal S(\bm{t}))$. In some situations, the procedures proposed might overestimate $V_m(\mathcal S(\bm{t}))$ and thus pay a price in power. For a graphical illustration of this phenomenon, we revisit Figure \ref{fig:mkn_illustration}. Note that we are using the number of points in the entire shaded left half-plane to estimate the number of false positives in just the shaded upper right quadrant. Unfortunately, this issue is not very easy to resolve. One challenge is that if we allow $\widehat V_m$ to depend on the entire vector $\bm{t}$, then the definition of $\bm{t}^*$ would be complicated by the fact that the lower left-hand corner property would no longer hold. Another challenge is that the dependencies between statistics across layers make it hard to come up with a better and yet tractable estimate of $V_m(\mathcal S(\bm t))$. Despite this flaw, the multilayer knockoff filter (and the p-filter) enjoys very similar power to its single-layer counterpart, as we shall see in the next section.

\section{Simulations} \label{sec:numerical}

We rely on simulations to explore the FDR control and the power of the multilayer knockoff filter and the generalized p-filter across a range of scenarios, designed to capture the variability described in the previous section. All code is available at \url{http://web.stanford.edu/~ekatsevi/software.html}. 
\subsection{Performance of the multilayer knockoff filter} \label{sec:mkn_simulations}

\paragraph{Simulation setup}

We simulate data  from the  linear model with $n>N$. This also allows us to calculate p-values for the null hypotheses $\beta_j=0$ and plug these into BH and p-filter; these two methods, in addition to KF, serve as points of comparison to MKF. 

We simulate
\begin{equation*}
\bm{y} = \bm{X}\bm{\beta} + \bm{\eps}, \quad \bm{\eps} \sim \mathcal N(\bm{0}, \bm{I}),
\end{equation*}
where $\bm{X}  \in \mathbb R^{n \times N}$, with $n = 4500$ observations on $N = 2000$ predictors. $\bm{X}\!$ is constructed by sampling each row independently from $N(\bm{0}, \bm{\Sigma}_\rho)$, where $(\bm\Sigma_\rho)_{ij} = \rho^{|i-j|}$ is the covariance matrix of an AR(1) process with correlation $\rho$. There are $M = 2$ layers: one comprising individual variables and one with $G = 200$ groups, each of  10 variables. The vector $\bm{\beta}$ has 75 nonzero entries. The indices of the non-null elements are determined by firstly selecting  $k$ groups uniformly at random, and then choosing, again uniformly at random, 75 elements of these  $k$ groups. Here, $k$ controls the strength of the group signal; we considered three values: low saturation ($k = 40$), medium saturation ($k = 20$), and high saturation ($k = 10$). We generated these three sparsity patterns of $\bm \beta$ once and fixed them across all simulations; see Figure \ref{fig:three_saturations}. In all cases, the nonzero entries of $\bm{\beta}$ are all  equal, with a magnitude that satisfies
\begin{equation*}
\text{SNR} = \frac{\norm{\bm{X}\bm{\beta}}^2}{n}
\end{equation*}
for a given SNR value. For each saturation setting, we vary $\rho \in \{0.1, 0.3, \dots, 0.9\}$ while keeping SNR fixed at 0.5, and vary $\text{SNR} \in \{0, 0.1, \dots, 0.5\}$ while keeping $\rho$ fixed at 0.3. Across all experiments, we used nominal FDR levels $q_{\text{ind}} = q_{\text{grp}} = 0.2$. 

This choice of simulation parameters captures some of the features of  genetic data: the AR(1) process for the rows of $\bm{X}\!$ is a first approximation  for the local spatial correlations of genotypes and the signal is relatively sparse, as we would expect in GWAS. A notable  difference between our simulations and  common genetic data is the scale: a typical GWAS involves $N\approx 1,000,000$ variables. Previous studies \citep{CetL16,SetC17} have already demonstrated the feasibility of knockoffs for datasets of this scale and the MKF does not appreciably differ in computational requirements. However, given our interest in exploring a variety of sparsity and saturation regimes, we found it convenient to rely on  a smaller scale. Moreover, working in a regime where $n>2N$ allows us to leverage the fixed design construction of knockoff variables, which does not require knowledge of the distribution of $\bm X$. 

\begin{figure}[t!]
	\centering
	\includegraphics[width = \textwidth]{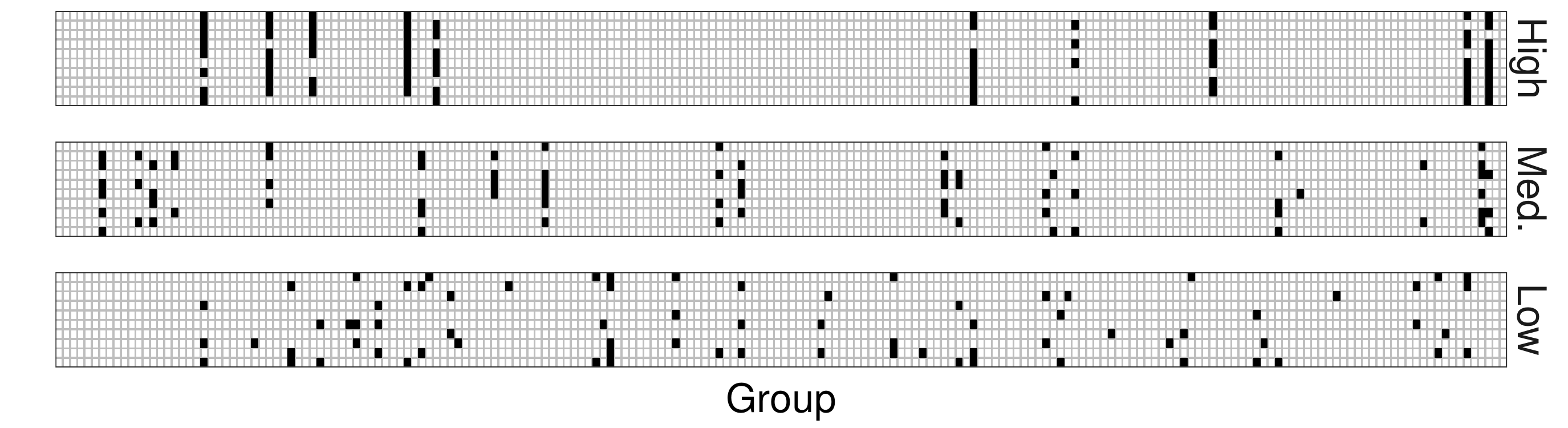}
	\caption{Simulated sparsity patterns for the three saturation regimes: each square corresponds to one variable, and each column to one group. Non null variables are indicated with filled squares.}
	\label{fig:three_saturations}
\end{figure}

\paragraph{Methods compared}
We compare the following four methods on this simulated data:
\begin{enumerate}
	\item[(a)] KF+ with fixed design knockoffs, lasso-based variable importance statistics combined using the signed-max function, targeting $q_{\text{ind}}$.
	\item[(b)] MKF+ with fixed design group knockoffs, ``Simes-like" group importance statistics based on the penalty $\ell_g^m(\bm{u}) = \norm{\bm{u}}_1$, combined using the signed-max function, targeting $q_{\text{ind}}$ and $q_{\text{grp}}$. We find that this choice of penalty has better power across a range of saturation levels than the group lasso based construction of \cite{DB16}. 
	\item[(c)] Benjamini Hochberg procedure (BH) on the p-values based on t-statistics from linear regression, targeting $q_{\text{ind}}$.
	\item[(d)] p-filter (PF) on the same set of p-values, targeting $q_{\text{ind}}$ and $q_{\text{grp}}$.
\end{enumerate}
Note that the first two methods are knockoff-based and the last two are p-value based, and that methods (a) and (c) target only the FDR at the individual layer while methods (b) and (d) target the FDR at both layers.

\paragraph{Results}

Figure \ref{fig:all_sim} illustrates our findings. 
First, consider the FDR of the four methods. The multilayer knockoff filter achieves multilayer FDR control across all parameter settings;  the constant $c_{\text{kn}} = 1.93$ from our proof does not appear to play a significant role in practice. The p-filter also has multilayer FDR control, even though the PRDS assumption is not satisfied by the two-sided p-values we are using. On the other hand, the knockoff filter and BH both violate FDR control  at the group layer as the saturation level and power increase. 

We also note that both the multilayer knockoff filter and regular knockoff filter have, on average, a realized FDP that is smaller than the target FDR. This is partly because we use the ``knockoffs+" version of these methods, which is conservative when the power is low. In addition, we find that the multilayer knockoff filter is conservative at the individual layer even in high-power situations if the saturation is high. This is a consequence of our construction of $\widehat V_m$, an estimate of the number of false discoveries that, as we have discussed, tends to be larger than needed. We see similar behavior for the p-filter, since it has an analogous construction of $\widehat V_m$. 

Next, we compare the power of the four methods. As expected, the power of all methods improves with  SNR  and degrades with $\rho$.  We find that the knockoff-based approaches consistently outperform the p-value based approaches, with higher power despite having lower FDRs and the  gap  widening as saturation increases. This power difference is likely caused by the ability of the knockoff-based approaches to leverage the sparsity of the problem to construct more powerful test statistics for each variable. Finally, we compare the power of the multilayer knockoff filter to that of the regular knockoff filter: in most cases, the multilayer knockoff filter loses little or no power, despite providing an additional FDR guarantee. This holds even in the low saturation setting, where the groups are not very informative for the signal.
%\begin{figure}[h]
%	\includegraphics[width = \textwidth]{figures/SNR_simulation.png}
%	\caption{Varying SNR.}
%	\label{fig:SNR_sim}
%\end{figure}
%\begin{figure}[h]
%	\includegraphics[width = \textwidth]{figures/rho_simulation.png}
%	\caption{Varying feature correlation.}
%	\label{fig:rho_sim}
%\end{figure}

\begin{figure}[t!]
\begin{center}
	\includegraphics[width = \textwidth]{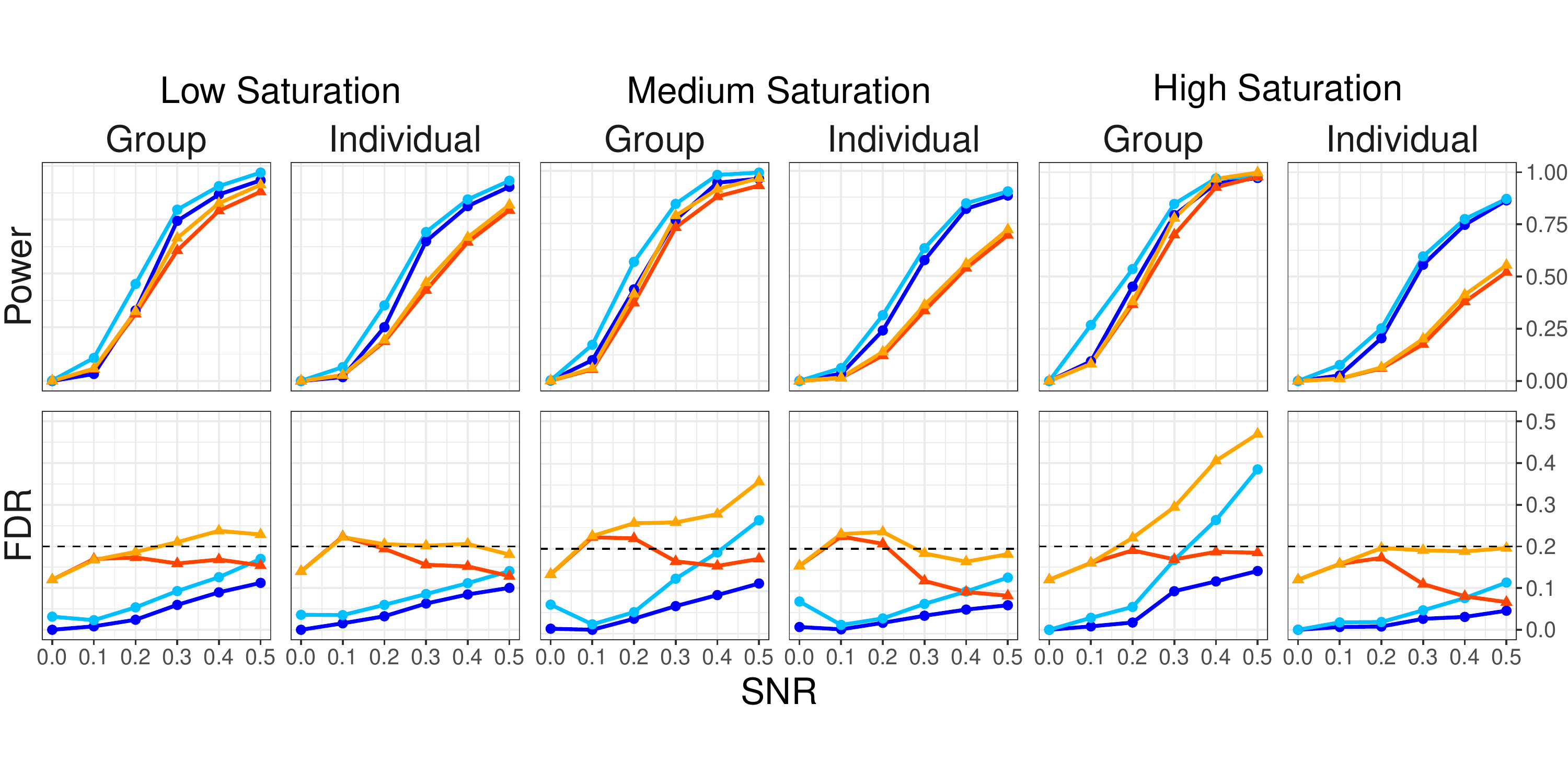}
	\vspace{-0.2in}
	
	\includegraphics[width =\textwidth]{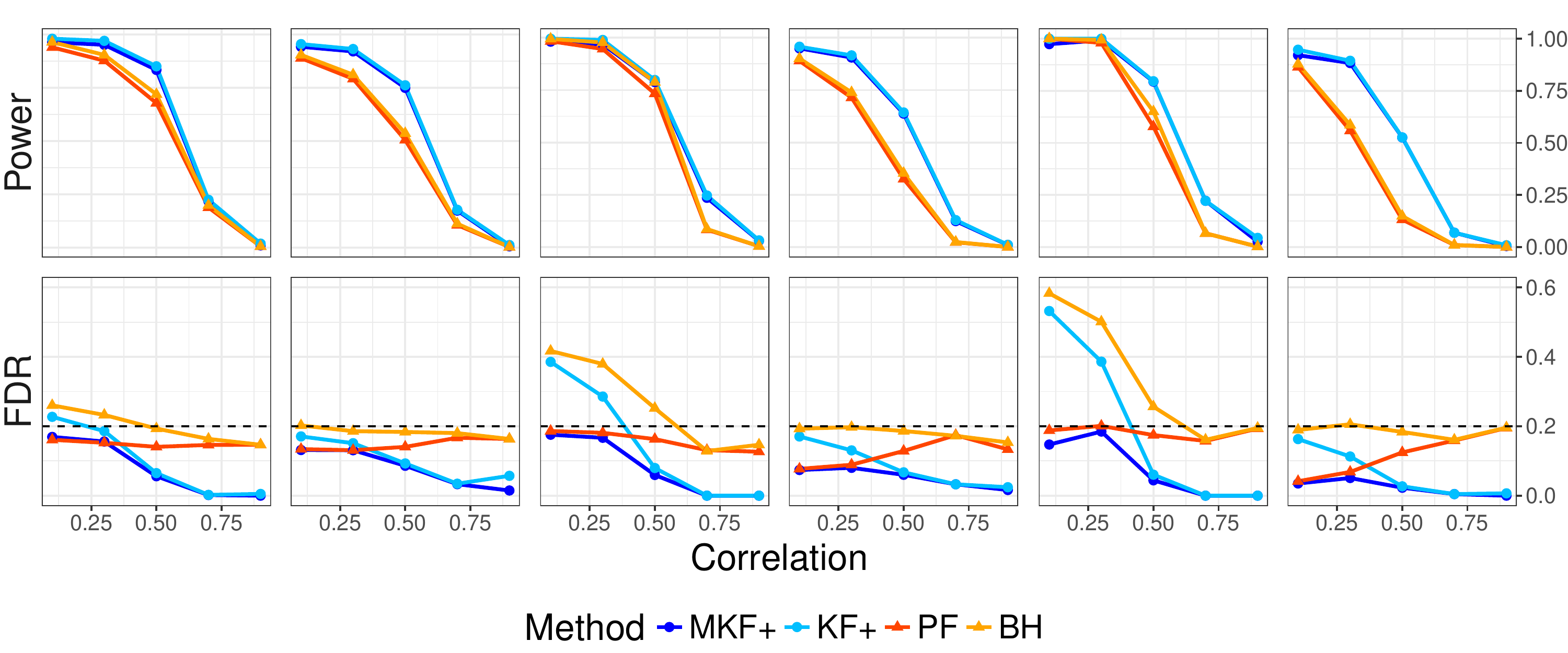}
	\end{center}
	\caption{Simulation results. From left to right, the saturation regime changes. The top panel varies signal-to-noise ratio while fixing $\rho = 0.3$. The bottom panel varies $\rho$ while fixing $\text{SNR} = 0.5$. Each point represents the average of 50 experiments.}
	\label{fig:all_sim}
\end{figure}

\subsection{Performance of the generalized p-filter}

We explore the possible advantages of the generalized p-filter in a setup when signals are expected to be weak and common  within non-null groups, so one would want to define group p-values via the Fisher test instead of the Simes test. We consider  two partitions of interest, both with   groups of size 10 (thus no singleton layer). % The two partitions  are staggered with respect to each other. 
A situation similar to this might arise when scientists are interested in determining which functional genomic segments are associated with a trait. There exist several algorithms to split the genome into functional blocks (e.g. ChromHMM by \cite{EK12}), and segments in each of these can be partially overlapping.

\paragraph{Simulation setup}
We simulated $N = 2000$ hypotheses, with $M = 2$ layers. Each layer had 200 groups, each of size 10. The groups in the second layer were offset from the those in the first layer by 5. Hence, the groups for layer one are $\{1, \dots, 10\}, \{11, \dots, 20\}, \dots$, while the groups for layer two are $\{6, \dots, 15\}, \{16, \dots, 25\}, \dots$. The nonzero entries of $\beta$ are $\{1, \dots, 200\}$. Hence, this is a ``fully saturated" configuration. 
We generate $X_j \overset{\text{ind}}\sim \mathcal N(\mu_j,1)$, where $\mu_j = 0$ for null $j$ and $\mu_j = \mu$ for non-null $j$. We then derive two-sided p-values based on the z test. In this context, we define $\text{SNR} = \norm{\mu}^2/N$. The SNR varied in the range $\{0, 0.1, \dots, 0.5\}$ and we targeted $q_{\text{ind}} = q_{\text{grp}} = 0.2$.

\paragraph{Methods compared}

\begin{itemize}
	\item[(a)] The regular p-filter, which is based on the Simes test.
	\item[(b)] The generalized p-filter with p-values based on the Fisher test.
\end{itemize}

\paragraph{Results}
Figure \ref{fig:fisher_v_simes}  shows how both versions of the generalized p-filter have multilayer FDR control, with the Fisher version  being more conservative. As with the multilayer knockoff filter, we see that the extra theoretical multiplicative factor is not necessary (at least in this simulation). In this case, Fisher has substantially higher power than Simes due to the weak distributed effects in each group. 
%This shows that the p-filter can be improved upon by choosing more powerful p-value constructions for groups. The generalized p-filter provides the analyst with the ability to use her domain knowledge to construct powerful group p-values at each layer. 

\begin{figure}[h]
	\centering
	\includegraphics[width = 0.7\textwidth]{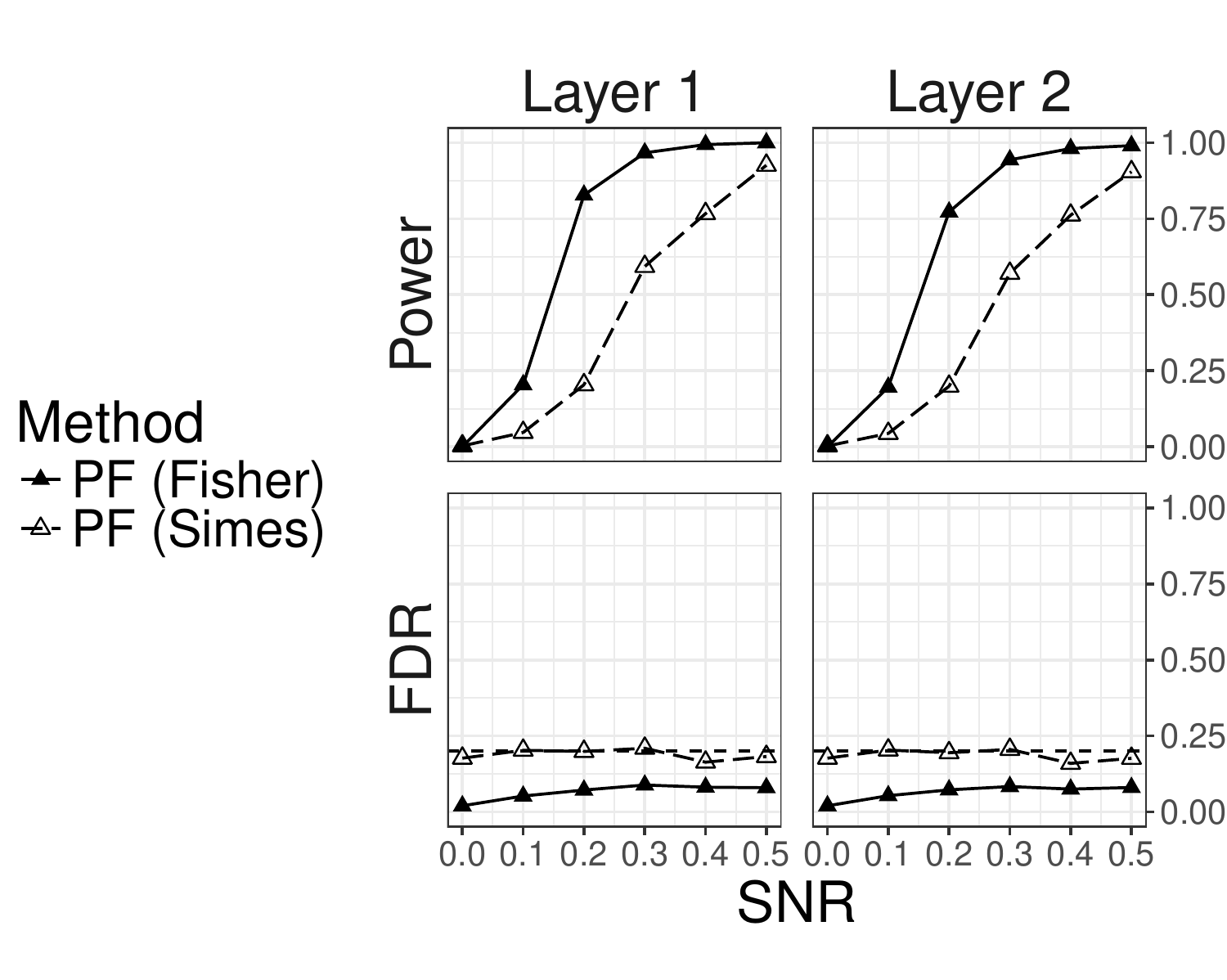}
	\caption{Performance of the generalized p-filter: comparison of Fisher and Simes combination rules.}
	\label{fig:fisher_v_simes}
\end{figure}

\section{Case study: variants and genes influencing HDL} \label{sec:real_data}

To understand which genes are involved in determining cholesterol levels and which genetic variants have an impact on its value, \cite{SetF14} carried out exome resequencing of 17 genetic loci  identified in previous GWAS studies as linked to metabolic traits in about 5000 subjects. The original analysis is based on marginal tests for common variants and burden tests for the cumulative effect of the rare variants in a gene \citep{LL08,  WetL11}. Furthermore, to account for linkage disequilibrium and estimate the number of variants that influence cholesterol in each location, \cite{SetF14} uses model selection based on BIC. The original analysis, therefore, reports findings at the gene level and variant level, but these findings derive from multiple separate analyses and lack coordination. Here we deploy MKF to leverage multiple regression models, obtaining a coherent set of findings  at both the variant and gene level, with approximate multilayer FDR control.

%Now, we test the multilayer knockoff filter on a genetic dataset. 

\paragraph{Data}  The resequencing targeted the coding portion of 17 genetic loci, distributed over 10 chromosomes and  containing  79 genes. This resulted in the identification of a total of 1304 variants. We preprocessed the data as in \cite{SS16}, who re-analyzed the data in a Bayesian framework. In particular, we removed variants with minor allele counts below a threshold, and pruned the set of polymorphisms to assure that the empirical correlation between any pair is of at most 0.3 (this is a necessary step in multiple regression analysis to avoid collinearity; see \cite{BetS17a} and \cite{CetL16}).
After preprocessing, the data contained 5335 individuals and 768 variants.
Since the study design was exome resequencing, every variant could be assigned to a gene. A special case is that  of 18  SNPs that were typed in a  previous study of these subjects and that were included in the final analysis as indicators of the original association signal; 12 of these  are located in coding regions, but 6 are not.

%See \cite{SS16} for details. 
While \cite{SetF14} studies the genetic basis of several metabolic traits, we focus our analysis here on HDL cholesterol. From the original measurement we regressed out  the effects of sex, age, and the first five principal components of genomewide genotypes, representing population structure.

Note that the small size of  the data set is due to the design of the study, which relies on targeted resequencing  as opposed to an exome-wide or genome-wide resequencing. Working with a small set of variables, that have been quite extensively analyzed already, allows us to better evaluate the MKF results, which is useful in a first application.  The analysis with MKF of a new exome-wide data set comprising tens of thousands of individuals and hundreds of thousands of variants is ongoing. 

%
%After preprocessing, the data contained 5335 individuals and 768 variants. Of these variants, 18 are SNPs measured in the original GWAS studies and 750 are variants obtained from resequencing. Of these 18 SNPs, 12 are located in coding regions and 6 are not. For each individual, we also have a continuous HDL cholesterol measurement, corrected for a few covariates such as age and sex.

\paragraph{Methods compared} To focus on the effect of adding  multilayer FDR guarantees (rather than on the consequences of different methods of analysis), we compare the results of  the multilayer knockoff filter (MKF) and the knockoff filter (KF). The multilayer knockoff filter used a variant layer and a gene layer. We chose MKF and KF instead of MKF+ and KF+ for increased power, but otherwise used the same method settings as in Section \ref{sec:numerical}. Each variant from the sequencing data and the 12 exonic GWAS SNPs were assigned to groups based on gene. The 6 intergenic GWAS SNPs are considered single members of 6 additional groups. Hence, our analysis has 85 (= 79 + 6) ``genes" in total. 

\paragraph{Results} 
Table \ref{table:results_summary} summarizes how many genes and SNPs each method discovers. KF has about twice as many discoveries at each layer, but how many of these are spurious? 
Unfortunately the identity of the variants truly associated with HDL is unknown, % Of course, we do not have access to the true genetic associations with HDL cholesterol. However, since HDL cholesterol is an intensively studied trait, 
but we can get an approximation to the truth using the existing literature and online databases. At the variant level, this task is difficult because (1) linkage disequilibrium (i.e. correlations between nearby variants) makes the problem ill-posed and (2) rare variants, present in this sequencing dataset, are less well studied and cataloged. Instead, we focus on an annotation at the gene level. Comparing the two methods at the gene level is also meaningful because this is the layer at which the multilayer knockoff filter provides an extra FDR guarantee. See Appendix \ref{sec:annotations} for references supporting our annotations.

Table \ref{table:gene_level_results} shows the gene layer results: % based on these annotations: 
there are 5 true positive genes (ABCA1, CETP, GALNT2, LIPC, LPL) found by both methods, 1 false positive shared by both methods (PTPRJ), 1 true positive for KF that is missed by MKF (APOA5), and 4 false positives (NLRC5, SLC12A3, DYNC2LI1, SPI1) for KF that MKF correctly does not select. Hence, MKF reduced the number of false positives from 5 to 1 at the cost of 1 false negative.

Figure \ref{fig:real_data} shows a more detailed version of these association results, illustrating the signal at the variant level.  Notably, the one extra false negative (APOA5) incurred by MKF just barely misses the cutoff for the gene layer. Aside from the extra false negative and the one false positive shared with KF, the additional horizontal cutoff induced by the need to control FDR at the gene level  does a good job separating the genes associated with HDL from those that are not. 
\begin{table}[h]
	\centering
	\begin{tabular}{ccc}
		\textbf{Method} & \textbf{\# SNPs found} & \textbf{\# Genes found} \\
		\hline
		KF & 23 & 11 \\
		MKF & 13 & 6 \\				
	\end{tabular}
	\caption{Summary of association results on resequencing data.}
	\label{table:results_summary}
\end{table}

\begin{table}[h]
	\centering
	\begin{tabular}{ccc}
		\textbf{Gene} & \textbf{Discovered by} & \textbf{Supported in literature}\\
		\hline
		ABCA1 & KF, MKF  & yes \\
		CETP & KF, MKF  & yes \\
		GALNT2 & KF, MKF  & yes \\
		LIPC & KF, MKF  & yes \\
		LPL & KF, MKF  & yes \\
		\textcolor{red}{PTPRJ} & KF, MKF & no \\
		APOA5 & KF & yes \\
		\textcolor{red}{NLRC5} & KF & no \\
		\textcolor{red}{SLC12A3} & KF & no \\
		\textcolor{red}{DYNC2LI1} & KF & no \\
		\textcolor{red}{SPI1} & KF & no
	\end{tabular}
	\caption{Comparison of MKF and KF at the gene layer. False positives are highlighted in red.
	}\label{table:gene_level_results}
\end{table}

\begin{figure}[h!]
	\centering
	\includegraphics[width = 0.85\textwidth]{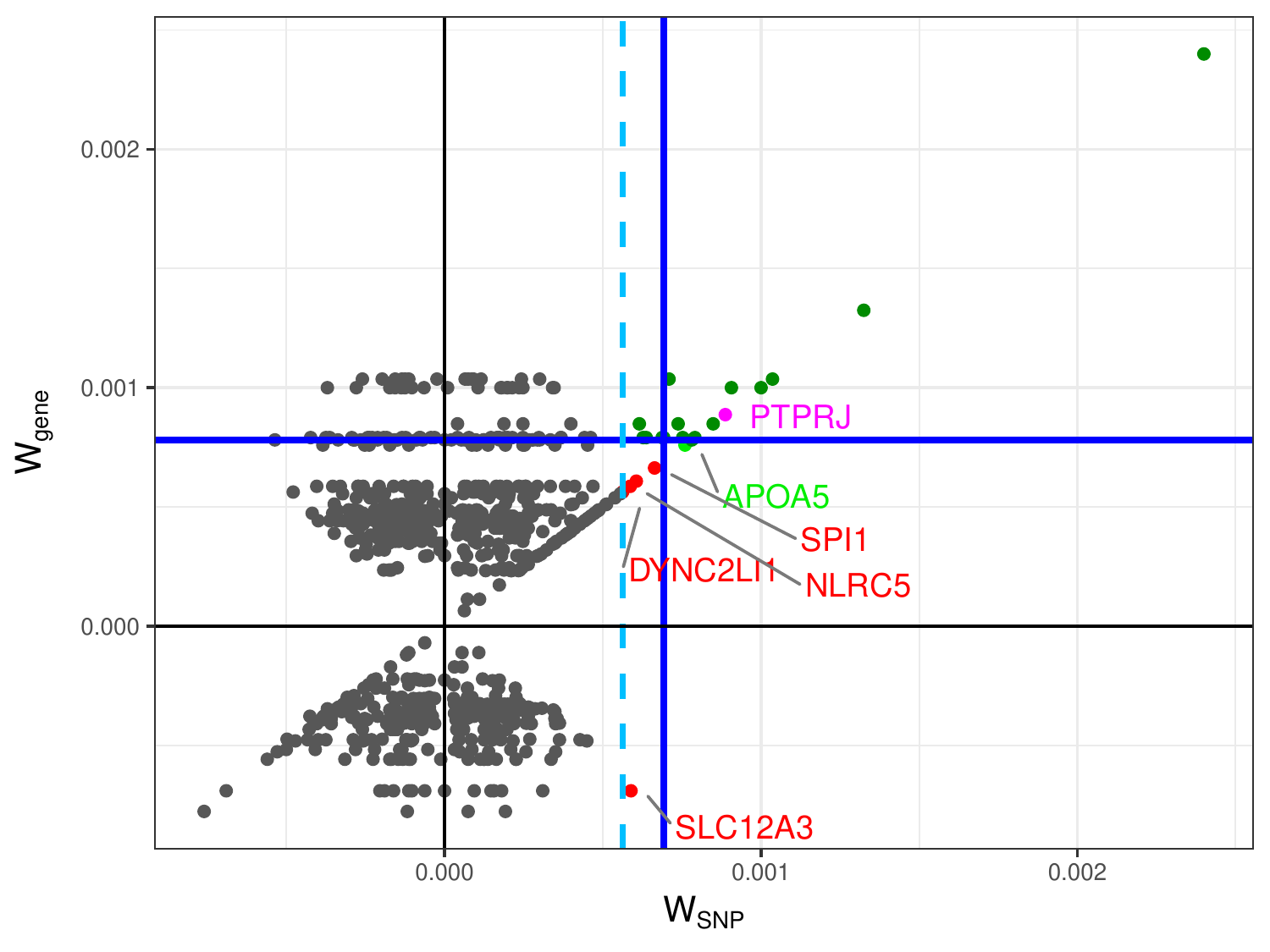}
	\caption{Scatterplot of variant level and gene level knockoff statistics. Solid blue lines are the thresholds for MKF while the cyan broken line is the threshold for KF. Each dot corresponds to a variant, and variants  selected by at least one of the methods are in color: dark green indicates selected variants that belong to genes that are true positives for both methods, light green a true positive found by KF but missed by MKF, red is for false positives of KF but true negatives of MKF, and magenta for the false positive shared by both methods. To facilitate comparison with Table \ref{table:gene_level_results} we indicate the names of the genes representing false positives or false negatives for at least one of the methods.}
	\label{fig:real_data}
\end{figure}

\section{Conclusions} \label{sec:conclusion}

With the multilayer knockoff filter, we have made a first step to equip model selection procedures with FDR guarantees for multiple types of reported discoveries, bridging results from the multi-resolution testing literature \citep{Y08, BB14, PetS16, BR15} with controlled selection methods \citep{BC15,DB16,CetL16}.
When tackling high dimensional variable selection, researchers have at their disposal several methods based on regularized regression, with penalties that can reflect an array of sparse structures (see for example \cite{KX09, JetP10, RetR13}), corresponding to a multiplicity of possible resolutions for discoveries. While several of them have been implemented in the context of genetic association studies 
\citep{ZetLa10,XetK14}, their application has been hampered by the lack of inferential guarantees on the selection. It is our hope that the approach put forward with MKF will allow scientists to leverage these computationally attractive methods to obtain replicable discoveries at multiple levels of granularity. 

%The multilayer knockoff filter represents a general paradigm for multilayer FDR control, inspired by the framework of the p-filter. This paradigm is very flexible and could potentially be applied to other kinds of feature statistics besides knockoff statistics and p-values. The multilayer knockoff filter is itself fairly general, and in particular is compatible with any valid constructions of knockoffs and knockoff statistics. Hence, any advances in this regard (like \cite{CetL16}) can automatically be incorporated into our method.

%The multilayer knockoff filter represents a first step towards providing more sophisticated FDR guarantees for model selection. While techniques like the sparse group lasso employ structured regularizers, the multilayer knockoff filter brings this kind of structure into the inferential guarantees of the method. Since scientific discoveries are routinely examined at multiple levels of resolution, having statistical methods with matching guarantees is essential to reproducibility of these discoveries.
%

In the process of developing a framework for multilayer FDR control for variable selection, we have also generalized the p-filter multiple testing procedure. Our approach places no restrictions on the relations between the p-values used to test the hypotheses at different layers. By contrast, theoretical results for the p-filter rely heavily on the specific way in which p-values for individual hypotheses are aggregated to obtain p-values for groups.  The constant $c_{\text{pf}}$ can be viewed as the price we pay for allowing these arbitrary dependencies. Nevertheless, simulations show that both $c_{\text{pf}}$ and the corresponding constant $c_{\text{kn}}$ for MKF appear to be inconsequential in practice.

Finally, the uniform bound (\ref{bounded_ratio_multilayer}) can be re-interpreted as stating that the maximum amount by which the FDP of the knockoff filter can exceed $\widehat{\text{FDP}}$ over the entire path has bounded expectation. Moreover, the same proof yields a high-probability uniform bound on FDP. This interpretation is pursued by \cite{KR18}, who also prove the corresponding bound for BH conjectured here, and extend the proof to a variety of other FDR methodologies.

\subsection*{Extensions}  

There are many directions in which it seems appropriate to extend the results we have so far. We list some below, of varying degrees of difficulty.

We have constructed  the multilayer framework so that it would not require the groups in different layers to be hierarchically nested: however, in certain applications, such a hierarchical structure exists and could be exploited to increase power, relaxing the consistency requirements we have for discoveries across layers. This is the case, for example, in the relation between SNPs and genes that partially motivated us: it is scientifically interesting  to discover that a gene is implicated in a disease, even if we are unable to pinpoint any specific causal variant within that gene.
Formally, 
  the selection sets allowed in this paper at each layer must be ``two-way consistent," i.e. selecting an individual variable implies selecting the group to which it belongs, and selecting a group implies selecting at least one variable in the group. In a hierarchical setting, a less stringent ``one-way consistency'' requirement can be formulated: selecting an individual variable implies selecting the group to which it belongs. In fact, it can be easily shown that MKF can be modified to enforce this relaxed consistency requirement, and the same proof technique shows that multilayer FDR control still holds. This modified MKF is currently being used to analyze data from an exome-wide resequencing study.%We plan to investigate the performance of this modified version of MKF in future work. 

%Next, we identify a few directions to expand on this work. The first such direction is to allow groups to overlap. 
In the genetic application motivating this work, SNPs are grouped according to the genes to which they belong. This is adequately described with non overlapping groups, but there are extensions in which it makes sense to consider groups that are overlapping within the same layer. For example, \textit{biological pathways} are groups of genes known to work together to carry out a certain biological function. It is often desired for inference to be carried out at the pathway level, as this gives a direct biological interpretation of the results. However, genes often participate in multiple pathways, which leads to large group overlaps. This brings new statistical challenges, starting from a meaningful description of the null hypothesis for each group, to the  construction of valid knockoff variables for overlapping groups. 

%Another direction for future work is to consider data-driven groupings instead of a priori groupings. 
We have focused on situations where the researcher,  prior to looking at the data,  can specify meaningful groups of variables corresponding to discoveries at coarser resolutions.  While this is the case in many settings where it makes sense to pursue FDR control at multiple layers, it is also true that there are problems where the groupings of predictors are most meaningful 
when based on the data. For example, in our own case study, we  grouped predictors as individual variables are too correlated with each other: while we did so in a fairly arbitrary manner and without looking at the outcome variable, choosing groups based on the data could  have helped us to choose a ``resolution" appropriate for the signal strength and correlation structure in the data. Even in a single-layer setup, selecting groups and then carrying out valid inference with respect to these groups is a challenging new problem: we hope that some of the ideas developed here can contribute to its solution.

%While in many cases meaningful groupings of predictors are known ahead of time, in other cases it would be interesting to learn these from the data. For example, it is common in single-cell RNA sequencing studies to identify new cell populations based on similar gene expression patterns. Alternatively, grouping might be done as a matter of necessity because individual predictors are too correlated with each other. While in this work and previous ones, this grouping was done in a somewhat arbitrary manner, choosing groups based on the data can help choose a ``resolution" appropriate for the signal strength and correlation structure in the data. Even in a single-layer setup, selecting groups and then carrying out valid inference with respect to these groups is a challenging new problem. 

Another promising extension of the MKF is to multi-task regression, the study of the impact of a set of predictor variables on multiple outcome variables. The multi-task regression problem is often reshaped into a larger single-task regression problem, in which the predictors have group structure based on which task they correspond to. For example, \cite{DB16} take this approach to multi-task regression alongside its development of the group knockoff filter. MKF can then provide a framework for FDR control in this setting, where group discoveries correspond to finding variables important for at least one of the outcomes, and individual discoveries correspond to the identification of variables important for a specific outcome. In the context of the linear model with $N \leq n$ and independent errors, the MKF as described here provides the desired FDR guarantee. However, the general case is more challenging and will require substantial modifications.

\section*{Acknowledgements} \label{sec:acknowledgements}

The authors are indebted to Emmanuel Cand\`{e}s and David Siegmund for help with theoretical aspects of this work. We also thank Lucas Janson for helpful comments on the manuscript, and editors and referee for constructive feedback. E.K. also thanks Subhabrata Sen for a helpful discussion. We thank the authors of \cite{SetF14} for giving as access to the data after their quality controls and acknowledge the dbGap data repository, studies phs000867 and phs000276.

\begin{supplement}
	\stitle{Proofs of technical results and evidence for gene annotations.}
	\slink[doi]{COMPLETED BY THE TYPESETTER}
	\sdatatype{.pdf}
	\sdescription{We provide proofs of technical results (in particular, Proposition \ref{prop:nondegeneracy} and lemmas supporting Theorems \ref{kn_multilayer_fdr_control} and \ref{th:pf_theorem}) and evidence for the gene annotations from Section \ref{sec:real_data}.}
\end{supplement}

\bibliographystyle{imsart-nameyear}
\bibliography{my_bib}

%%%%%%%%%%%%%%%%%%%%%%%%%%%%%%%%%%%%%%%%

\appendix

\section{Connection between group and individual null hypotheses}
\begin{proof}[Proof of Proposition \ref{prop:nondegeneracy}]
	Fix $\mathcal A = \mathcal A_g^m$ for some $g, m$, and assume without loss of generality that $\mathcal A = \{1, \dots, K\}$. It is clear that if $Y \independent X_{\mathcal A} | X_{-\mathcal A}$, then $Y \independent X_{j} | X_{-j}$ for each $j \in \mathcal A$. Hence, the left-hand side  of equation (\ref{equivalence}) from the main text is included in the right-hand side.  \color{black}Conversely, suppose that $Y \independent X_{j} | X_{-j}$ for each $j \in \mathcal A$. Fix $y$ and $x, x' \in \mathcal D$. Then, we have
	\begin{equation*}
	\begin{split}
	p_{Y | X_{\mathcal A}, X_{-\mathcal A}}(y|x_{\mathcal A}, x_{-\mathcal A}) &= p_{Y|X_{\mathcal A \setminus \{1\}}, X_{-\mathcal A}}(y | x_{\mathcal A \setminus \{1\}}, x_{-\mathcal A}) \\
	&= p_{Y | X_{\mathcal A}, X_{-\mathcal A}}(y|(x'_1, x_2, \dots, x_K), x_{-\mathcal A}).
	\end{split}
	\end{equation*}
	Both equalities follows from the conditional independence assumption $Y \independent X_{1} | X_{-1}$, and all conditional distributions are well-defined (i.e., we are not conditioning on an event of probability zero) because we assumed the entire domain has nonzero probability density. Proceeding in this way for variables $2, \dots, K$, we find that
	\begin{equation*}
	p_{Y | X_{\mathcal A}, X_{-\mathcal A}}(y|x_{\mathcal A}, x_{-\mathcal A}) = p_{Y | X_{\mathcal A}, X_{-\mathcal A}}(y|x'_{\mathcal A}, x_{-\mathcal A}).
	\end{equation*}
	Since this holds for any $y, x, x'$, this implies that $Y \independent X_{\mathcal A} | X_{-\mathcal A}$. This proves the reverse inclusion, so we are done. 
\end{proof}
Note that this result and proof technique are strongly related to Lemma 3.2 in \cite{CetL16}.

\section{FDR control for multilayer knockoff filter} \label{sec:kn_appendix}

For clarity, we begin by proving the sufficient condition (\ref{sufficient}) for Theorem \ref{kn_multilayer_fdr_control} from the main text with the slightly worse constant of 2.1. 
\begin{lemma} \label{lemma:kn_bound_looser}
	Let $S_k = \sum_{i = 1}^k X_i$, where $X_i \overset{i.i.d.}\sim \text{Ber}(1/2)$ and $S_0 = 0$. Then,
	\begin{equation*}
	a_{\text{\em kn}} = \mathbb E\left[\sup_{k\geq 0}\frac{S_k}{1 + k - S_k}\right] \leq 2.1.
	\end{equation*}
\end{lemma}

This statement reduces to bounding the probability the random walk $S_k$ hits a linear boundary.

\begin{lemma} \label{random_walk_linear_boundary}
	Fix constants $c_1 \in (\frac12, 1)$ and $c_2$. We have
	\begin{equation*}
	\mathbb P[S_k \geq c_1 k + c_2 \text{ for some } k \geq 0] \leq \exp(-\Theta(c_1) \cdot c_2),
	\end{equation*}
	where $\Theta(c_1)$ is the unique positive root of
	\begin{equation}
	\exp(\theta(1-c_1)) + \exp(-\theta c_1) = 2.
	\label{theta_eq_kn}
	\end{equation}
\end{lemma}
Indeed, assuming the statement of Lemma \ref{random_walk_linear_boundary}, we write
\begin{equation}
\begin{split}
a_{\text{kn}} &= \mathbb E\left[\sup_{k \geq 0}\frac{S_k}{1 + k - S_k}\right] \\
&= \int_0^\infty \mathbb P \left[\sup_{k \geq 0}\frac{S_k}{1 + k - S_k} \geq t\right]dt \\
&= 1 + \int_1^\infty \mathbb P \left[\sup_{k \geq 0}\frac{S_k}{1 + k - S_k} \geq t\right]dt \\
&= 1 + \int_1^\infty \mathbb P \left[S_k \geq \frac{t}{1+t}k + \frac{t}{1+t} \text{ for some } k \geq 0\right]dt \\
&\leq 1 + \int_1^\infty \exp\left(-\Theta\left(\frac{t}{1+t}\right) \cdot \frac{t}{1+t}\right) dt \\
&\approx 2.0998... \\
&\leq 2.1.
\label{expectation_as_integral}
\end{split}
\end{equation}
The second line follows from the first because $S_k/(1 + k - S_k) \overset{a.s.}\rightarrow 1$ by the law of large numbers. The fourth line follows by Lemma \ref{random_walk_linear_boundary} because $\frac{t}{1+t} \in (\frac12, 1)$ for $t > 1$, and the fifth line comes from numerical integration. Although the integrand is not technically defined at $t = 1$, note that we might as can set $\Theta(1/2) = 0$, for which the statement in Lemma \ref{random_walk_linear_boundary} holds trivially, leading to an integrand of 1 at $t = 1$.

\begin{proof}[Proof of Lemma \ref{random_walk_linear_boundary}]
	For fixed $\theta > 0$, define the process
	\begin{equation*}
	Z_k = \exp(\theta(S_k - c_1 k - c_2)).
	\end{equation*}
	Then, 
	\begin{equation*}
	\mathbb P[S_k \geq c_1 k + c_2 \text{ for some } k \geq 0] = \mathbb P\left[\sup_{k \geq 0} Z_k \geq 1\right].
	\end{equation*}	
	The goal is to choose $\theta$ so that $Z_k$ becomes a martingale with respect to the filtration $\mathcal F_k = \sigma(\{X_i\}_{1 \leq i \leq k})$. For this, it is sufficient that
	\begin{equation*}
	\begin{split}
	1 = \mathbb E\left[\left.\frac{Z_k}{Z_{k-1}}\right|\mathcal F_{k-1}\right] &= \mathbb E[\exp(\theta(X_k - c_1))|\mathcal F_{k-1}] \\
	&= \mathbb E[\exp(\theta(X_k - c_1))] = \frac12\exp(\theta(1 - c_1)) + \frac12 \exp(-\theta c_1).
	\label{OST}
	\end{split}
	\end{equation*}
	We claim that this equation has a unique positive root for $c_1 \in (\frac12, 1)$. Indeed, let $f(\theta) = \frac12\exp(\theta(1 - c_1)) + \frac12 \exp(-\theta c_1)$. Observe that $f(0) = 1$, $f'(0) = \frac12 - c_1 < 0$, and $\lim_{\theta \rightarrow \infty} f(\theta) = \infty$. these facts, together with the convexity of $f$, imply that it has a unique positive root $\Theta(c_1)$. 
	
	Finally, by applying Doob's maximal inequality for nonnegative martingales as well as the monotone convergence theorem, we obtain
	\begin{equation}
	\mathbb P\left[\sup_{k \geq 0} Z_k \geq 1\right] \leq \mathbb E[Z_0] = \exp(-\Theta(c_1) \cdot c_2).	
	\label{ineq_loose}
	\end{equation}
\end{proof}

We can do better than Lemma \ref{lemma:kn_bound_looser}, because the inequality in (\ref{ineq_loose}) (and hence the constant 2.1) is somewhat loose. This is related to the fact that the random walk $S_k$ sometimes overshoots the linear boundary. This overshoot can be fairly substantial at the beginning of the random walk. To get the improved constant of 1.93, as stated in Theorem \ref{kn_multilayer_fdr_control}, we carry out a more refined analysis.

\begin{lemma} \label{bound_knockoffs}
	Define $a_{\text{kn}}$ as in Lemma \ref{lemma:kn_bound_looser}. Then, $a_{kn} \leq 1.93$.
\end{lemma}

\begin{proof}
	In order to remedy this overshoot, we consider all $2^{k_0}$ possibilities for the first $k_0$  steps of the random walk, and then carry out the same analysis as above starting at time $k_0$. Fix $k_0 \geq 0$. Then, 
	\begin{equation*}
	\begin{split}
	&\mathbb P \left[\sup_{k \geq 0}\frac{S_k}{1 + k - S_k} \geq t\right] \\
	&\quad= \mathbb E\left[\mathbb P \left[\left.\sup_{k \geq 0}\frac{S_k}{1 + k - S_k} \geq t\right|X_1, \dots, X_{k_0}\right]\right] \\
	&\quad= 2^{-k_0}\sum_{x_1, \dots, x_{k_0}}\mathbb P \left[\left.\sup_{k \geq 0}\frac{S_k}{1 + k - S_k} \geq t\right|X_1 = x_1, \dots, X_{k_0} = x_{k_0}\right].
	\end{split}
	\end{equation*}
	Next, define
	\begin{equation*}
	R_{k_0}(x) = \max_{k \leq k_0} \frac{\sum_{i \leq k}x_i}{1 + k - \sum_{i \leq k}x_i}; \quad P_{k_0}(x) = \sum_{i \leq k_0} x_i.
	\end{equation*}
	These are the maximum ratio and partial sum after $k_0$ steps, respectively. Then, we have
	\begin{equation*}
	\begin{split}
	&\mathbb P \left[\sup_{k \geq 0}\frac{S_k}{1 + k - S_k} \geq t\right] \\
	&\quad= 2^{-k_0}\sum_{x_1, \dots, x_{k_0}}\left\{I(R_{k_0}(x) \geq t) + \vphantom{\mathbb P \left[\sup_{k \geq 0}\frac{P_{k_0}(x) + \overline{S}_k}{1 + k_0 + k - (P_{k_0}(x) + \overline S_k)} \geq t\right]}\right. \\
	&\qquad \qquad \qquad \qquad \left.\mathbb P \left[\sup_{k \geq 0}\frac{P_{k_0}(x) + \overline{S}_k}{1 + k_0 + k - (P_{k_0}(x) + \overline S_k)} \geq t\right] I(R_{k_0}(x) < t)\right\},
	\end{split}
	\end{equation*}
	where
	\begin{equation*}
	\overline S_k = \sum_{i = k_0+1}^{k_0 + k} X_i.
	\end{equation*}
	Note that $\{\overline S_k\}_{k \geq 0}$ has the same distribution as $\{S_k\}_{k \geq 0}$, since $X_i$ are i.i.d. Integrating the above expression and using an argument similar to  (\ref{expectation_as_integral}), we get
	\begin{equation*}
	\begin{split}
	a_{\text{kn}} &= 2^{-k_0} \sum_{x_1, \dots, x_{k_0}}\left\{\max(R_{k_0}(x),1) + \vphantom{\int_{\max(R_{k_0}(x),1)}^{\infty}}\right. \\
	&\qquad \qquad \qquad \quad \left. \int_{\max(R_{k_0}(x),1)}^{\infty}\mathbb P \left[\sup_{k \geq 0}\frac{\overline S_k + P_{k_0}(x)}{1 + k_0 + k - (\overline S_k+P_{k_0}(x))} \geq t\right]dt \right\}.
	\end{split}
	\label{first_step_analysis}
	\end{equation*}
	Note that 
	\begin{equation*}
	\frac{\overline S_k + P_{k_0}(x)}{1 + k_0 + k - (\overline S_k+P_{k_0}(x))} \geq t \quad \Longleftrightarrow \quad \overline S_k \geq \frac{t}{1+t}k + \frac{t}{1+t}(1 + k_0) - P_{k_0}(x).
	\end{equation*}
	Hence, by Lemma \ref{random_walk_linear_boundary}, we have
	\begin{equation*}
	\begin{split}
	&\mathbb P \left[\sup_{k \geq 0}\frac{\overline S_k + P_{k_0}(x)}{1 + k_0 + k - (\overline S_k+P_{k_0}(x))} \geq t\right] \\
	&\quad \leq \exp\left(-\Theta\left( \frac{t}{1+t}\right)\cdot \left(\frac{t}{1+t}(1+k_0) - P_{k_0}(x)\right)\right).
	\end{split}
	\end{equation*}
	Hence, we have
	\begin{equation*}
	\begin{split}
	a_{\text{kn}} &\leq 2^{-k_0} \sum_{x_1, \dots, x_{k_0}}\left\{\max(R_{k_0}(x),1) + \vphantom{\int_{\max(R_{k_0}(x),1)}^{\infty}}\right.\\
	&\qquad \quad \left. \int_{\max(R_{k_0}(x),1)}^{\infty}\exp\left(-\Theta\left( \frac{t}{1+t}\right)\cdot \left(\frac{t}{1+t}(1+k_0) - P_{k_0}(x)\right)\right)dt \right\}.
	\end{split}
	\label{interim}
	\end{equation*}
	Recall that that $\Theta(t/(1+t))$ is defined as the root of the nonlinear equation 
	\begin{equation*}
	\exp(\theta/(1+t)) + \exp(-\theta t/(1+t)) = 2.
	\end{equation*}
	Define
	\begin{equation*}
	\theta_t = \frac{1}{1+t}\Theta\left(\frac{t}{1+t}\right).
	\end{equation*}
	Then, $\theta_t$ satisfies
	\begin{equation}
	e^{\theta_t} + e^{-t\theta_t} = 2.
	\label{theta_eq}	
	\end{equation}
	In terms of $\theta_t$, we get
	\begin{equation}
	\begin{split}
	a_{\text{kn}}&\leq 2^{-k_0} \sum_{x_1, \dots, x_{k_0}}\left\{\max(R_{k_0}(x),1) + \vphantom{\int_{\max(R_{k_0}(x),1)}^{\infty}}\right.\\
	&\qquad \quad \left. \int_{\max(R_{k_0}(x),1)}^{\infty}\exp\left((1+k_0 - P_{k_0}(x))\log(2 - e^{\theta_t}) + \theta_t P_{k_0}(x)\right) dt \right\}.
	\end{split}
	\label{interim_2}
	\end{equation}
	Solving (\ref{theta_eq}) for $t$ in terms of $\theta_t$, we get
	\begin{equation*}
	t = \frac{-\log(2 - e^{\theta_t})}{\theta_t}.
	\end{equation*}
	At this point, we change variables in (\ref{interim_2}). Taking a derivative, the transformation is
	\begin{equation*}
	dt = \frac{\frac{e^{\theta_t}}{2 - e^{\theta_t}}\theta_t + \log(2 - e^{\theta_t})}{\theta_t^2}d\theta_t.
	\end{equation*}
	This yields
	\begin{equation}
	\begin{split}
	&\int_{\max(R_{k_0}(x),1)}^{\infty}\exp\left((1+k_0 - P_{k_0}(x))\log(2 - e^{\theta_t}) + \theta_t P_{k_0}(x)\right) dt \\
	&\quad= \int^{\log 2}_{\theta_{\max(R_{k_0}(x),1)}} \exp\left((1+k_0 - P_{k_0}(x))\log(2 - e^{\theta_t}) + \theta_t P_{k_0}(x)\right)\frac{\frac{e^{\theta_t}}{2 - e^{\theta_t}}\theta_t + \log(2 - e^{\theta_t})}{\theta_t^2}d\theta_t \\
	&\quad= \int^{\log 2}_{\theta_{\max(R_{k_0}(x),1)}} \exp\left((k_0 - P_{k_0}(x))\log(2 - e^{\theta_t}) + \theta_t P_{k_0}(x)\right)\frac{e^{\theta_t}\theta_t + (2 - e^{\theta_t})\log(2 - e^{\theta_t})}{\theta_t^2}d\theta_t \\
	\end{split}
	\label{change_of_variables}
	\end{equation}
	Putting together (\ref{interim_2}) and (\ref{change_of_variables}), we can compute the desired bound using only numerical integration over bounded intervals.
	
	\begin{figure}[h]
		\centering
		\includegraphics[width = \textwidth]{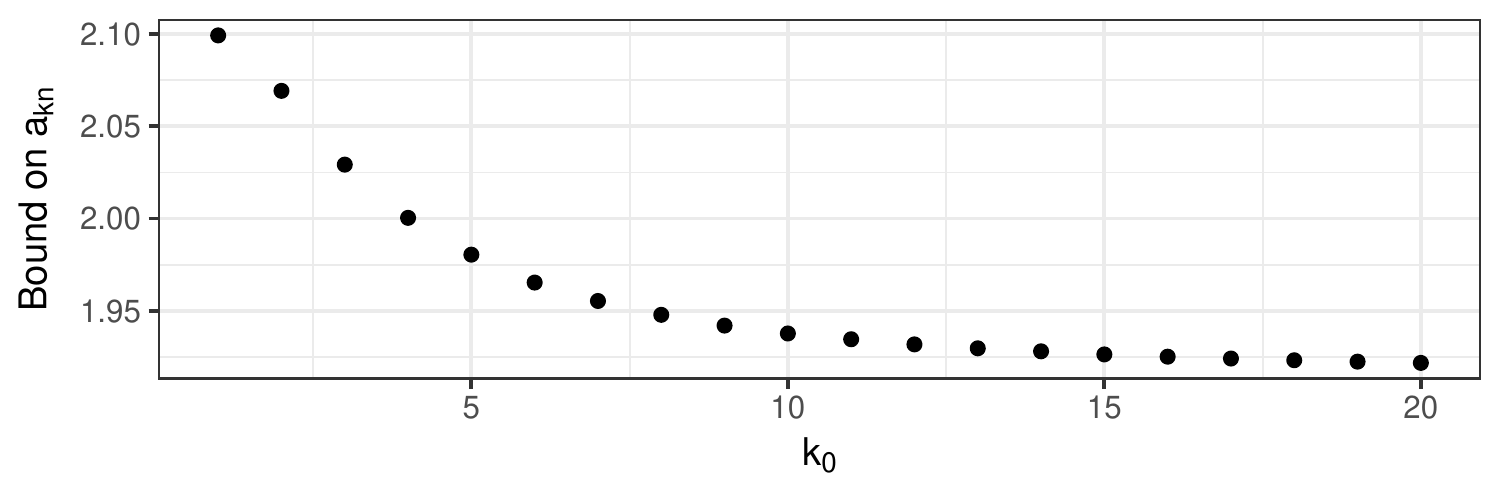}
		\caption{Upper bounds on $a_{\text{kn}}$.}
		\label{upper_bounds}
	\end{figure}
	
	Figure \ref{upper_bounds} shows the bounds we obtain from this approach for $k_0 = 1, \dots, 20$. For instance, the bound obtained using $k_0 = 20$ is 1.922, which proves the lemma.
\end{proof}

\section{FDR control for generalized p-filter} \label{sec:pf_appendix}

For the sake of readability, the notation in this self-contained section departs somewhat from the notation in the main text of the paper. In particular, $n$ will replace $G$ to denote the number of p-values, and in this section does not represent the sample size.

\begin{lemma}\label{bound_pfilter}
	Let $F_n(t)$ be the empirical CDF of $n$ independent uniform random variables. Then
	\begin{equation}
	a_{\text{\em pf}}(n) = \mathbb E\left[\sup_{t \in [0, 1]}\frac{F_n(t)}{n^{-1} + t}\right] \leq 1 + \exp\left(n^{-1/2} + \frac12 n^{-1}\right)0.42 + en^{-1/4}.
	\end{equation}
\end{lemma}

Note that $N_t = nF_n(t) \sim \mu$, where $\mu$ is the distribution of a Poisson process with rate $n$, conditioned on $N_1 = n$. Hence, we may equivalently bound the quantity
\begin{equation*}
\mathbb E_\mu\left[\sup_{t \in [0, 1]}\frac{N_t}{1 + nt}\right].
\end{equation*}
We first state two auxiliary lemmas, and show how the main result follows.
\begin{lemma} \label{I}
	Let $t_0 = n^{-1/2}$. Then, 
	\begin{equation*}
	\mu\left[\sup_{t \in [0, t_0]}\frac{N_t}{1 + nt} \geq x\right] \leq \exp\left(n^{-1/2} + \frac12 n^{-1}\right)\exp(-\gamma_x x),
	\end{equation*}
	where $\gamma_x$ is the positive root of
	\begin{equation*}
	e^{\gamma} = 1 + \gamma x.
	\end{equation*}
	
\end{lemma}

\begin{lemma} \label{II}
	Let $t_0 = n^{-1/2}$. Then, 
	\begin{equation*}
	\mu\left[\sup_{t \in [t_0, 1]}\frac{N_t}{1 + nt} \geq x\right] \leq \exp(1 - n^{1/4}(x-1)).
	\end{equation*}
\end{lemma}

\begin{proof}[Proof of Lemma \ref{bound_pfilter}]
	
	Taking $t_0 = n^{-1/2}$ and using the results of Lemmas \ref{I} and \ref{II}, we find 
	\begin{equation*}
	\begin{split}
	a_{\text{pf}}(n) &= \mathbb E_\mu\left[\sup_{t \in [0, 1]}\frac{N_t}{1 + nt}\right] \\
	&= \int_0^\infty \mu\left[\sup_{t \in [0, 1]}\frac{N_t}{1 + nt} \geq x\right]dx \\
	&\leq 1 + \int_1^\infty \mu\left[\sup_{t \in [0, 1]}\frac{N_t}{1 + nt} \geq x\right]dx \\
	&\leq 1 + \int_1^\infty \mu\left[\sup_{t \in [0, t_0]}\frac{N_t}{1 + nt} \geq x\right]dx + \int_1^\infty \mu\left[\sup_{t \in [t_0, 1]}\frac{N_t}{1 + nt} \geq x\right]dx \\
	&\leq 1 + \int_1^\infty \exp\left(n^{-1/2} + \frac12 n^{-1}\right)\exp(-\gamma_x x) dx + \int_1^\infty  \exp(1 - n^{1/4}(x-1)) dx \\
	&= 1 + \exp\left(n^{-1/2} + \frac12 n^{-1}\right)0.42 + en^{-1/4}.
	\end{split}
	\end{equation*}
	Here, we use the fact (obtained by numerical integration), that $\int_1^\infty \exp(-\gamma_x x)dx = 0.42$.
\end{proof}
Now, we prove the two auxiliary lemmas.
\begin{proof}[Proof of Lemma \ref{I}]	
	Let $\nu$ be the unconditional distribution of a Poisson process with rate $n$. The idea is that $\nu$ is easier to work with than $\mu$ due to the martingale properties of the unconditional Poisson process, and $\nu$ and $\mu$ are close on $[0, t_0]$. Hence, we change measure from $\mu$ to $\nu$:
	\begin{equation}
	\mu\left[\sup_{t \in [0, t_0]}\frac{N_t}{1 + nt} \geq x\right] = \mathbb E_\nu\left[\frac{d\mu}{d\nu}I\left(\sup_{t \in [0, t_0]}\frac{N_t}{1 + nt} \geq x\right)\right].
	\label{change_of_measure}
	\end{equation}
	Our first claim is that
	\begin{equation}
	\text{ess sup}_{\nu}\ \frac{d\mu}{d\nu} \leq C_n = \exp\left(n^{-1/2} + \frac12 n^{-1}\right),
	\label{ess_sup}
	\end{equation}
	where $\mu$ and $\nu$ are viewed as measures on the space of stochastic processes on $[0, t_0]$ (that are right continuous with left limits).	Let $\mathcal F_t = \sigma(N_s: s \leq t)$. It suffices to show that 
	\begin{equation}
	\mu[A] \leq C_n \nu[A] \quad \text{for all } A \in \mathcal F_{t_0}.
	\label{probability_bound}
	\end{equation}
	To show this, let
	\begin{equation*}
	\mathcal C = \{A \in \mathcal F_{t_0}: \mu[A] \leq C_n \nu[A]\},
	\end{equation*}
	and let
	\begin{equation*}
	\mathcal R = \{\{N_{t_1} \in B_1, \dots, N_{t_k} \in B_k\}: 0 \leq t_1 < \cdots < t_k \leq t_0, k \geq 1, B_1, \dots, B_k \subset \mathbb N \}.
	\end{equation*}
	be the collection of finite-dimensional rectangles, which generates $\mathcal F_{t_0}$. Let $\mathcal A$ be the algebra generated by $\mathcal R$, i.e. the collection of finite disjoint unions of sets in $\mathcal S$. We claim that it suffices to show that $\mathcal R \subset \mathcal C$. Indeed, if $\mathcal R \subset \mathcal C$, then by additivity it is clear that $\mathcal A \subset \mathcal C$. But since $\mathcal A$ is an algebra and $\mathcal C$ is a monotone class (due to the continuity of measures), the monotone class theorem implies that $\mathcal F_{t_0} = \sigma(\mathcal A) \subset \mathcal C$, from which it follows that $\mathcal C = \mathcal F_{t_0}$, which is the statement (\ref{probability_bound}). Finally, to show that $\mathcal R \subset \mathcal C$, by countable additivity it suffices to show that $\{N_{t_1} = y_1, \dots, N_{t_k} = y_k\} \in \mathcal C$ for all $y_1, \dots, y_k \in \mathbb N$. Note that
	\begin{equation*}
	\begin{split}
	\frac{\mu[N_{t_1} = y_1, \dots, N_{t_k} = y_k]}{\nu[N_{t_1} = y_1, \dots, N_{t_k} = y_k]} &= \frac{\mu[N_{t_k} = y_k]\mu[N_{t_1} = y_1, \dots, N_{t_{k-1}} = y_{k-1} | N_{t_k} = y_k]}{\nu[N_{t_k} = y_k]\nu[N_{t_1} = y_1, \dots, N_{t_{k-1}} = y_{k-1} | N_{t_k} = y_k]} \\
	&= \frac{\mu[N_{t_k} = y_k]}{\nu[N_{t_k} = y_k]},
	\end{split}
	\end{equation*}
	where the conditional distributions in the numerator and denominator cancel due to the Markov property of Poisson processes. Hence, it suffices to show that for all $t \leq t_0$ and all $y \in \mathbb N$,
	\begin{equation*}
	\frac{\mu[N_t = y]}{\nu[N_t = y]} \leq C_n.
	\end{equation*}
	Note that $\mu[N_t = y]= 0$ for $y > n$. For $y \leq n$, we write
	\begin{equation*}
	\begin{split}
	\log\left(\frac{\mu[N_t = y]}{\nu[N_t = y]}\right) &= \log\left(\frac{{n \choose y}t^{y}(1 - t)^{n - y}}{e^{-nt}\frac{(nt)^{y}}{y!}}\right) \\
	&= \log\left(\frac{n!}{n^{y}(n - y)!}(1-t)^{n - y}e^{nt}\right) \\
	&=\log\left(\prod_{i = 1}^{y - 1}\left(1 - \frac i n\right)\right) + (n - y)\log(1 - t) + nt \\
	&= \sum_{i = 1}^{y- 1}\log\left(1 - \frac i n\right) + (n - y)\log(1 - t) + nt \\
	&\leq -\sum_{i = 1}^{y - 1}\frac{i}{n} -(n - y)\left(t + \frac12 t^2\right) + nt \\
	&= -\frac1{2n}y (y- 1) -(n - y)\left(t + \frac12 t^2\right) + nt \\
	&= -\frac{1}{2n}y^2 + \left(\frac{1}{2n} + t + \frac12 t^2\right)y -\frac12 nt^2. 
	\end{split}
	\end{equation*}
	The value $y_* = \frac12 + nt + \frac12 nt^2$ maximizes the above expression, so we get
	\begin{equation*}
	\begin{split}
	\log\left(\frac{\mu[N_t = y]}{\nu[N_t = y]}\right) &\leq -\frac{1}{2n}y_*^2 + \left(\frac{1}{2n} + t + \frac12 t^2\right)y_* -\frac12 nt^2 \\
	&= \frac n 2 \left(\frac 1 {2n} + t + \frac12 t^2\right)^2 - \frac12 nt^2 \\
	&= \frac{1}{8n} + \frac{nt^4}{8} + \frac t 2 + \frac{t^2}{4} + \frac{t^3 n}{2} \\
	&\leq \frac{1}{8n} + \frac{nt_0^4}{8} + \frac {t_0} 2 + \frac{t_0^2}{4} + \frac{t_0^3 n}{2} \\
	&= \frac{1}{8n} + \frac{1}{8n} + \frac 1 2 n^{-1/2} + \frac{1}{4n} + \frac{1}{2}n^{-1/2} = n^{-1/2} + \frac12 n^{-1}.
	\end{split}
	\end{equation*}
	This proves statement (\ref{ess_sup}). Together with (\ref{change_of_measure}), this implies that
	\begin{equation}
	\mu\left[\sup_{t \in [0, t_0]}\frac{N_t}{1 + nt} \geq x\right] \leq  \exp\left(n^{-1/2} + \frac12 n^{-1}\right)\nu\left[\sup_{t \in [0, t_0]}\frac{N_t}{1 + nt} \geq x\right].
	\label{intermediate}
	\end{equation}
	Hence, we can now compute the probability under the measure $\nu$. Let 
	\begin{equation*}
	Z_t = \exp(\gamma_x N_t - nt(e^{\gamma_x} - 1)),
	\end{equation*}
	where $\gamma_x$ is defined as in the statement of Lemma \ref{I}. Under $\nu$, this is an exponential martingale associated with the Poisson process $N_t$. We have
	\begin{equation*}
	\begin{split}
	\nu\left[\sup_{t \in [0, t_0]}\frac{N_t}{1 + nt} \geq x\right] &\leq \nu\left[\sup_{t \geq 0} \frac{N_t}{1 + nt} \geq x\right] \\
	&= \nu\left[\sup_{t \geq 0}(N_t - (1+nt)x) \geq 0\right] \\
	&= \nu\left[\sup_{t \geq 0}\left(\exp(\gamma_x N_t - \gamma_x (1+nt)x)\right) \geq 1\right] \\
	&= \nu\left[\sup_{t \geq 0}\left(\exp(\gamma_x N_t - nt(e^{\gamma_x}-1) - \gamma_x x)\right) \geq 1\right] \\
	&= \nu\left[\sup_{t \geq 0} Z_t \geq \exp(\gamma_x x)\right] \\
	&\leq \exp(-\gamma_x x)\mathbb E[Z_0] = \exp(-\gamma_x x),
	\end{split}
	\end{equation*}
	where the last inequality follows from the maximal inequality for martingales.
	Putting this together with (\ref{intermediate}) completes the proof of the lemma.
\end{proof}

\begin{proof}[Proof of Lemma \ref{II}]
	
	For any $\alpha > 0$, we have
	\begin{equation*}
	\mu\left[\sup_{t \in [t_0, 1]}\frac{N_t}{1 + nt} \geq x\right] \leq \mu\left[\sup_{t \in [t_0, 1]}\frac{N_t}{nt} \geq x\right] \leq \mu\left[\sup_{t \in [t_0, 1]}\exp\left(\alpha\frac{N_t}{nt}\right) \geq \exp(\alpha x)\right].
	\end{equation*}
	Now, since $N_t/nt$ is a backwards martingale under $\mu$, it follows that $\exp(\alpha N_t/nt)$ is a backwards submartingale, so from the maximal inequality we have
	\begin{equation*}
	\begin{split}
	\mu\left[\sup_{t \in [t_0, 1]}\exp\left(\alpha\frac{N_t}{nt}\right)\geq \exp(\alpha x)\right]
	&\leq \exp(-\alpha x)\mathbb E_\mu\left[\exp\left(\alpha \frac{N_{t_0}}{nt_0}\right)\right] \\
	&= \exp(-\alpha x)\mathbb E_\mu\left[\exp\left(\sum_{j = 1}^n \frac{\alpha I(U_j \leq t_0)}{nt_0}\right)\right]\\
	&= \exp(-\alpha x) \left((1-t_0) + t_0 \exp(\alpha/nt_0)\right)^n \\
	&= \exp(-\alpha x) \left(1 + t_0 (\exp(\alpha/nt_0)-1)\right)^n. \\
	\end{split}
	\end{equation*}
	Here, $U_1, \dots, U_n$ are independent uniform random variables. At this stage, let us take $\alpha = n^{1/4}$ and $t_0 = n^{-1/2}$. Then, the bound becomes
	\begin{equation*}
	\exp(-n^{1/4}x) \left(1 + n^{-1/2} (\exp(n^{-1/4})-1)\right)^n
	\end{equation*}
	Note that $\exp(n^{-1/4})-1 \leq n^{-1/4} + \frac12 n^{-1/2}\exp(n^{-1/4}) \leq n^{-1/4} + n^{-1/2}$ for $n \geq 5$. Hence,
	\begin{equation*}
	\begin{split}
	\mu\left[\sup_{t \in [t_0, 1]}\frac{N_t}{1 + nt} \geq x\right] &\leq \exp(-n^{1/4}x) \left(1 + n^{-1/2} (n^{-1/4} + n^{-1/2})\right)^n \\
	&\leq \exp(-n^{1/4}x)\exp(n^{1/4} + 1) \\
	&= \exp(1-n^{1/4}(x-1)).
	\end{split}
	\end{equation*}
	This completes the proof of the lemma.	
\end{proof}

\begin{remark} \label{monotonicity}
	Using a formula from \cite{KT81} (Chapter 13, problems 49 and 51) for the probability of an empirical process hitting a linear boundary, we can actually get the following exact expression for any $x > 1$:
	\begin{equation*}
	\begin{split}
	&\mu\left[\sup_{t \in [0, 1]}\frac{N_t}{1 + nt} \geq x\right] \\
	&\quad= \sum_{i = 0}^{\lfloor n - x \rfloor} {n \choose i}\left(\frac{n - i - x}{nx}\right)^{n-i} \left(1 - \frac{n - i - x}{nx}\right)^{i - 1}\left(1 + \frac{1}{n} - \frac 1 x\right).
	\label{exact}
	\end{split}
	\end{equation*}
	Hence, we find that
	\begin{equation}
	\begin{split}
	a_{\text{kn}} &= \mathbb E_\mu\left[\sup_{t \in [0, 1]}\frac{N_t}{1 + nt} \geq x\right] \\
	&= 1 + \int_1^\infty \left\{\sum_{i = 0}^{\lfloor n - x \rfloor} {n \choose i}\left(\frac{n - i - x}{nx}\right)^{n-i} \left(1 - \frac{n - i - x}{nx}\right)^{i - 1}\left(1 + \frac{1}{n} - \frac 1 x\right)\right\}dx.
	\label{KT_exact}
	\end{split}
	\end{equation}
	While this quantity is hard to work with theoretically, we may compute it. Figure \ref{bounds} shows this quantity, along with the limiting value (equal to 1.42) from Lemma \ref{bound_pfilter}. This plot strongly suggests that $a_{\text{\em pf}}(n)$ is an increasing function of $n$, which leads us to conjecture that $a_{\text{\em pf}}(n) \leq \limsup_{n \rightarrow \infty} a_{\text{\em pf}}(n) \leq 1.42$ for all $n$. 
	
\end{remark}

\begin{figure}[h]
	\centering
	\includegraphics[width = 0.7\textwidth]{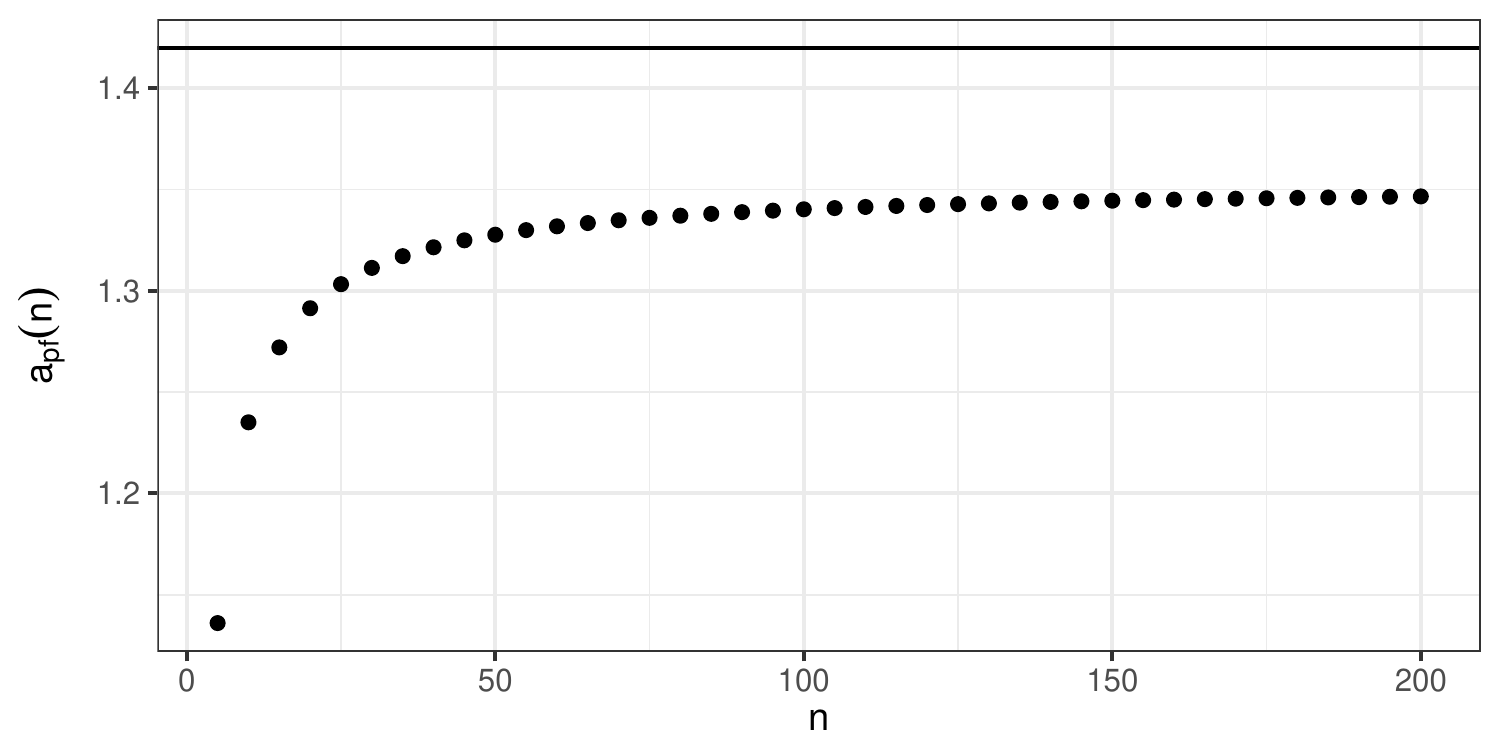}
	\caption{Points represent the exact expression (\ref{KT_exact}), and the horizontal line represents the limiting bound from Lemma \ref{bound_pfilter}.}
	\label{bounds}
\end{figure}

\section{Evidence for gene annotations} \label{sec:annotations}

In this section, we discuss the evidence in literature and databases for our annotations of association or no association with HDL cholesterol. We use the Online Mendelian Inheritance in Man (OMIM) database \citep{HetM05} and NHGRI/EBI GWAS catalog \citep{WetH14}.

\paragraph{Likely Associated with HDL}

\begin{itemize}
	\item The genes ABCA1, CETP, GALNT2, LIPC, LPL are all well-known to be associated with HDL. These associations are documented in OMIM and in the GWAS catalog. 
	
	\item We also conclude that APOA5 is likely associated with HDL, although the primary trait for which this gene is known is triglycerides. This association with HDL is documented in the GWAS catalog and in \cite{LetC16}. 
\end{itemize}

\paragraph{Likely not associated with HDL}

\begin{itemize}
	\item The genes PTPRJ, DYNC2LI1, and SPI1 show no evidence in literature or databases of association with HDL.
	\item The gene NLRC5 has no reported association with HDL in OMIM or the GWAS catalog. There is one paper (Charlesworth et al 2009) that predicts an association using a gene-based analysis. However, this gene is very near CETP. Moreover, this paper states that ``Interestingly, the list also prioritizes a number of genes of little-known function, such as NLRC5 (NLR family CARD domain containing 5)...,which would not be selected by any form of candidate gene approach." 
	\item SLC12A3 is a gene in the CETP locus that does contain a GWAS hit. However, a paper reporting this association \citep{RetC09} states that SLC12A3 is not an obvious candidate for association with HDL and hypothesizes that the effect is ``mediated by long range linkage disequilibrium to a causal variant nearer the CETP gene." 
\end{itemize}

\end{document}